\documentclass[lettersize,journal]{IEEEtran}
\usepackage{cite}
\usepackage{amsmath,amsthm,amssymb,,amsfonts}
\usepackage{algorithmic}
\usepackage{textcomp}

\usepackage[ruled,vlined]{algorithm2e}
\usepackage{dsfont}
\usepackage{mathrsfs}
\usepackage{mathbbol}

\usepackage{accents}

\usepackage{float}
\usepackage{comment}

\usepackage{xcolor,graphicx,epstopdf}
\usepackage{dsfont}
\usepackage{hyperref}
\usepackage{caption}
\usepackage{subcaption}

\usepackage[subfigure]{tocloft}

\usepackage{gensymb}

\usepackage{textcomp,mathcomp}
\newtheorem{theorem}{Theorem}
\newtheorem{definition}{Definition}

\newtheorem{remark}{Remark}
\newtheorem{proposition}{Proposition}
\newtheorem{assumption}{Assumption}

\begin{document}
\title{Peer-to-Peer Energy Markets With Uniform Pricing: A Dynamic Operating Envelope Approach}
\author{Zeinab Salehi, Yijun Chen, Ian R. Petersen, Guodong Shi, Duncan S. Callaway, and Elizabeth L. Ratnam
\thanks{Zeinab Salehi and Ian R. Petersen are with the School of Engineering, Australian National University, Canberra, ACT 2601, Australia (e-mail:  zeinab.salehi@anu.edu.au; ian.petersen@anu.edu.au).}

\thanks{Yijun Chen is with the Department of Electrical and Electronic Engineering,  University of Melbourne,  VIC 3052, Australia (e-mail:  yijun.chen.1@unimelb.edu.au).}

\thanks{Guodong Shi is with the  Australian Centre for Robotics, School of Aerospace, Mechanical and Mechatronic Engineering, The University of Sydney, NSW 2050, Australia (e-mail: guodong.shi@sydney.edu.au).}

\thanks{Duncan S. Callaway is with the Energy and Resources Group, University of California, Berkeley, CA 94720 USA (e-mail: dcal@berkeley.edu).}

\thanks{Elizabeth L. Ratnam is with the Department of Electrical and Computer Systems Engineering, Monash University, Melbourne, VIC 3168, Australia (e-mail:  liz.ratnam@monash.edu).}

}


\maketitle

\begin{abstract}
The recent widespread adoption of rooftop solar backed by battery storage is enabling energy customers to both produce and consume electricity (i.e., prosumers of electricity). To facilitate prosumer participation in the electric grid, new market mechanisms are required. In this paper, we design peer-to-peer energy markets where prosumers trade their excess energy with peers to gain profit while satisfying the overall balance in electricity supply and demand. We first consider a market structure, considering the case where voltage and/or thermal constraints are binding. When such grid constraints are binding, market clearing prices can vary across locations.  However, heterogeneous prices may be considered by regulators to lack fairness. To ensure uniform pricing, we design two peer-to-peer energy markets with dynamic operating envelopes (DOEs). DOEs enable us to decompose global voltage and thermal constraints across the power grid into local constraints for each prosumer, resulting in uniform prices across the grid. By means of numerical simulations on an IEEE 13-node feeder, we benchmark the proposed market-based approaches in the presence of binding voltage constraints.
\end{abstract}

\begin{IEEEkeywords}
Dynamic operating envelope (DOE), energy market, uniform pricing, binding grid constraints.
\end{IEEEkeywords}

\section{Introduction}
\label{sec:introduction}

\IEEEPARstart{T}{he} recent rapid proliferation of distributed energy resources, such as rooftop solar, battery storage, and electric vehicles, are transforming traditional electricity markets. In this new era of renewable energy, end-users are becoming active participants, known as prosumers, who not only consume energy but also produce energy and share it through electric distribution grids \cite{azim2023dynamic}. As electric distribution grids transition from a system with unidirectional power flows to a bidirectional energy system, innovative economic models that facilitate active prosumer participation are required.

In the recent literature, peer-to-peer (P2P) energy markets have been proposed to enable prosumers to profit from trading electricity \cite{nasiri2023moment, alfaverh2023dynamic, gao2023blockchain}. A P2P electricity market must maintain a real-time balance between supply and demand for reliable grid operations \cite{weron2014electricity}. 
It is expected that the price of electricity in a P2P market influences prosumer decisions regarding when to export electricity generation and when to increase electricity consumption. To balance supply and demand, the electricity price can be determined through the notion of a \textit{competitive equilibrium} in a competitive market \cite{mas1995microeconomic, nguyen2011walrasian, wu2023competitive, salehi2022finite}. 

In a competitive market, the impact of a (single) participant decision is typically small relative to the overall market --- meaning a prosumer decision to import or export electricity does not significantly influence market-based electricity prices \cite{mas1995microeconomic}. A competitive equilibrium is achieved when no participant has a financial incentive to unilaterally change their decision, and total demand matches total supply at every moment \cite{li2020transactive, salehi2023competitive, chen2022trilevel, salehi2022infinite}. Once electricity prices corresponding to a competitive equilibrium are communicated to each prosumer, each prosumer can seek to maximize their payoff independently. The payoff combines the prosumer income from transactions and their utility from consumption. 


Beyond balancing the supply and demand of electricity at all moments, electricity systems must also adhere to the physics-based limits of voltage and thermal constraints. When voltage or thermal constraints are binding, market participants can be exposed to different electricity prices based to their physical location in the electric grid, or otherwise termed,  \textit{locational prices} \cite{ meng2023transmission, li2015market, Li2021}. Locational pricing reflecting transmission grid constraints has been implemented in several countries, with the United States being the primary adopter \cite{litvinov2010design}. Recently, there has been growing interest in extending the concept of locational pricing to distribution grids to support market-based participation of distributed energy resources \cite{papavasiliou2017analysis, Morstyn2020, salehi2024acc}. However, locational price disparities have the potential to exacerbate  social inequities \cite{tangeras2021competitive, edmunds2017distributed, CitiPower}, influencing social acceptance and potentially driving customer dissatisfaction. Designing markets with uniform pricing for all participants, while ensuring grid constraints are satisfied, is a way to enhancing the social acceptance of market-based prices for prosumers in distribution grids.



In this paper, we design competitive energy markets with uniform electricity prices for electric distribution grids, accounting for power flow constraints such as voltage and thermal limits. We consider the concept of dynamic operating envelopes (DOEs)\footnote{A DOE defines a range of real and reactive power that each prosumer can inject into or draw from the electric grid without violating operational constraints \cite{mahmoodi2023capacity}.}, which enables us to decompose grid constraints into localized constraints that each prosumer manages independently. The main contributions of this paper are threefold.
 \begin{enumerate}
     \item We propose competitive markets with uniform pricing using the concept of  DOEs. We demonstrate that the resulting competitive equilibria maximize social welfare and vice versa. Social welfare is defined as the sum of all prosumer utilities as determined by a central coordinator in the absence of market-based prices.
     \item We show our proposed uniform pricing approach is strongly budget-balanced whereas locational pricing generates a budget surplus. That is, total payments collected by a market-based coordinator sum to zero with our uniform pricing approach.
     \item We formulate a relationship between the proposed uniform pricing market and the traditional competitive market with locational pricing.
 \end{enumerate}

The rest of the paper is organized as follows. Section \ref{sec:preliminaries} presents our preliminaries, including definitions of competitive markets with locational pricing and the concept of DOEs. In Section \ref{sec:competitive}, we design two competitive markets with uniform pricing based on DOEs. We then explore the connection between our proposed uniform pricing approaches and traditional locational pricing. Section \ref{sec:Numerical Results} provides numerical simulations, and Section \ref{sec:conclusion} concludes the paper.
 
\section{Preliminaries} \label{sec:preliminaries}

\subsection{Prosumer  Model}\label{section:prosumer_model}
Consider an electric distribution grid that is represented by a radial topology connecting $N$ prosumers. Each prosumer $i \in \mathcal{N}$, where $\mathcal{N}=\{1, 2, \cdots, N\}$, represents a building (or a group of buildings) with uncontrollable loads such as lights,  controllable loads such as EVs, and local electricity production from distributed energy resources (DERs) such as solar panels. The system operates over a finite time horizon which is divided into $T$ time intervals, each with a length of $\Delta$. Time intervals $t$ are indexed in the set $\mathcal{T}=\{ 0, 1, \dots , T-1 \}$. The net supply of prosumer $i$, which is the difference between the local electricity production from solar panels and the uncontrollable loads, is denoted by $a_i(t) \in \mathbb{R}$.  
The loads of each prosumer $i$ that can be actuated (i.e., controlled) are associated with dynamics represented by the linear difference equation
\begin{equation}\label{eq_state}
	\mathbf x_i(t+1) = \mathbf A_i \mathbf x_i(t) + \mathbf B_i \mathbf u_i(t), \quad t \in \mathcal{T},
\end{equation}
where $\mathbf x_i(t) \in \mathbb{R}^{n}$ is the dynamical state (e.g., the state of charge (SoC) of an EV battery), $\mathbf x_i(0) \in \mathbb{R}^{n}$ is the initial state (e.g., the initial SoC), and $\mathbf u_i(t) \in \mathbb{R}^{m}$ is the control action (e.g., charge/discharge rate of an EV battery). In addition, $\mathbf A_i \in \mathbb{R}^{n \times n}$ and $\mathbf B_i \in \mathbb{R}^{n \times m}$ are fixed matrices. The state and control inputs are subject to physical constraints 
$	\underline{\mathbf x}_i \leq \mathbf x_i(t) \leq \overline{\mathbf x}_i$ and $  \underline{\mathbf u}_i\leq \mathbf u_i(t) \leq \overline{\mathbf u}_i$, 
where $\underline{\mathbf x}_i \in \mathbb{R}^{n}$ and $\underline{\mathbf u}_i \in \mathbb{R}^{m}$ represent lower bounds, and $\overline{\mathbf x}_i \in \mathbb{R}^{n}$ and $\overline{\mathbf u}_i \in \mathbb{R}^{m}$  represent upper bounds. 
Each prosumer $i$ has a preference regarding their power consumption. Such preferences are encoded in a utility function $f_i(\mathbf x_i(t), \mathbf u_i(t)): \mathbb{R}^{n} \times \mathbb{R}^{m} \mapsto \mathbb{R}$ that represents the satisfaction of prosumer $i$ as a result of reaching the state 
 $\mathbf x_i(t)$ and taking the control action $\mathbf u_i(t)$ at time step $t \in \mathcal{T}$. For example, a prosumer is potentially satisfied and achieves a high $f_i(\cdot)$ if the SoC $\mathbf x_i(t)$ of an EV battery is close to $80 \%$ of the battery capacity, or the charge/discharge cycles are relatively low. The satisfaction of the prosumer at the final time step $T$ is denoted by the terminal utility function $\phi_i(\mathbf x_i(T)): \mathbb{R}^{n} \mapsto \mathbb{R}$ that depends only on the terminal state $\mathbf x_i(T)$.
The energy consumed or withdrawn by the controllable loads for taking the control action $\mathbf u_i(t)$ is denoted by  $h_i(\mathbf u_i(t)): \mathbb{R}^{m} \mapsto \mathbb{R}$. Each prosumer can trade their excess supply $a_i(t) - h_i(\mathbf u_i(t))$ through the grid by a trading decision variable $p_i(t) \in \mathbb{R}$ that represents the  average active power injection at node $i$ and time interval $t \in \mathcal{T}$, such that $p_i(t) \leq a_i(t) - h_i(\mathbf u_i(t))$.

\subsection{Power Grid  Model}\label{section:microgrid_model}
We consider a radial distribution grid with  $N+1$ nodes indexed in the set $\mathcal{N}\cup \{0\}$ and operates in an islanded mode. Each prosumer is connected to a node $i \in \mathcal{N}$, while node $0$ represents the reference node or feeder. The net average active and reactive power injections at each node $i \in \mathcal{N}\cup \{0\}$ and time interval $t \in \mathcal{T}$ are denoted by $p_i(t) \in \mathbb{R}$ and $q_i(t) \in \mathbb{R}$, respectively. In an islanded mode, there is no power injection from the bulk grid, meaning that $p_0(t)=q_0(t)=0$ for $t \in \mathcal{T}$. The voltage phasor at each node $i \in \mathcal{N}\cup \{0\}$ and time step $t$ is denoted by $V_i(t) \in \mathbb{C}$, and its squared magnitude is denoted by $v_i(t)=|V_i(t)|^2$. The voltage at the feeder is given and fixed such that $v_0(t)=v_0$ for $t \in \mathcal{T}$. 
Denote by $\mathcal{E}$ the set of all lines in the grid, where  $(i, j) \in \mathcal{E}$ denotes the line connecting nodes $i$ and $j$, characterized by  resistance $r_{ij}$ and reactance $ \chi_{ij}$.
Let $P_{ij}(t)$ and $Q_{ij}(t)$ denote the average active and reactive power flows from node $i$ to $j$, respectively. The power flow in a steady state and in an islanded radial grid can be approximately represented by the LinDistFlow model \cite{baran1989optimal}
\begin{equation}\label{eq1}
	\begin{aligned}
		&P_{ij}(t) = -p_j(t) + \sum_{k:(j,k) \in \mathcal{E}} P_{jk}(t),\\
		&Q_{ij}(t) = -q_j(t) + \sum_{k:(j,k) \in \mathcal{E}} Q_{jk}(t),\\
		&v_i (t)- v_j(t)= 2\big(r_{ij}P_{ij}(t) + \chi_{ij}Q_{ij}(t) \big).
	\end{aligned}
\end{equation}
In the remainder, by power we mean average power. Equation \eqref{eq1} implies that all injected power must be balanced at each time step, i.e., 
$	\sum_{i=1}^{N} p_i(t)=0$ and $
 \sum_{i=1}^{N} q_i(t)=0$ for  $t \in \mathcal{T}$.

\subsection{Power Grid Constraints}
Let $\mathbf p(t)=(p_1(t), \dots, p_N(t))^\top$ and $\mathbf q(t)=(q_1(t), \dots, q_N(t))^\top$ denote the vectors of active and reactive power injections at all nodes over time interval $t \in \mathcal{T}$, respectively. In this paper, we consider grid constraints in the separable form \cite{mahmoodi2023capacity}
\begin{equation}\label{eq:constraints}
    \mathbf F_t(\mathbf p(t), \mathbf q(t))=\sum_{i=1}^N \mathbf g_{it}( p_i(t),  q_i(t)) \leq \boldsymbol \nu(t), \quad t \in \mathcal{T},
\end{equation}
where $\mathbf F_t: \mathbb{R}^{N} \times \mathbb{R}^{N} \mapsto \mathbb{R}^M$ is an affine vector function, and $M$ is the total number of grid constraints at time step $t \in \mathcal{T}$. In addition, $\mathbf g_{it}: \mathbb{R}\times \mathbb R \mapsto \mathbb{R}^M$ is a local and affine vector function associated with prosumer $i$ at time step $t \in \mathcal{T}$ that represents the role of the prosumer in forming the $M$ global grid constraints. The characteristics of these functions depend on the grid properties, such as line impedances and the grid topology. The vector $\boldsymbol \nu(t) \in \mathbb{R}^M$ represents the bound on the grid constraints which is given and known from the grid physical and operational limits. Two examples of the grid constraints that can be written in the form of \eqref{eq:constraints} are voltage and thermal constraints \cite{mahmoodi2023capacity}. In the following we describe voltage constraints in detail.

Denote by $\mathcal{P}_i$  the set of lines on the unique path from node $0$ to node $i$. Let 
$R_{ik}=2\sum_{(h, l) \in \mathcal{P}_i \cap \mathcal{P}_k} r_{hl}$, and $X_{ik}=2\sum_{(h, l) \in \mathcal{P}_i \cap \mathcal{P}_k} \chi_{hl}$. Note that $R_{ik}$ and $X_{ik}$ are the sensitivity matrices under no-load conditions that satisfy 
$R_{ik}=\frac{\partial v_i(t)}{\partial p_k(t)}, X_{ik}=\frac{\partial v_i(t)}{\partial q_k(t)}$. The voltage equation in \eqref{eq1} can be written in terms of the active and reactive power injections such that 
	$v_i (t)=  v_0 +\sum_{k=1}^{N} \big(R_{ik} p_k (t)+  X_{ik}  q_k (t)\big)$ for $i \in \mathcal{N}$ \cite{farivar2013equilibrium}.
The voltage magnitude at each node $i \in \mathcal{N}$ must stay within an acceptable range, typically within $\pm5 \% $ of a nominal voltage. Such voltage constraints can be written as 
\begin{equation}\label{eq_voltage}
	\underline {v}_i \leq v_0 +\sum_{k=1}^{N} \big(R_{ik} p_k (t)+  X_{ik}  q_k (t)\big) \leq \overline { v}_i, \quad i \in \mathcal{N},
\end{equation}
where $\underline { v}_i$ and $\overline { v}_i$ are the lower and upper bounds on the squared voltage magnitude. It is straightforward to write \eqref{eq_voltage} in the vector form \eqref{eq:constraints}.

In the rest of this section, we review related works in the literature on energy markets with locational (nonuniform) prices in electric grids, as well as the concept of DOEs which are required for further developments. Let $\mathbf p_i=( p_i(0), \dots,  p_i(T-1))^\top$, $\mathbf q_i=( q_i(0), \dots,  q_i(T-1))^\top$, and $\mathbf U_i=(\mathbf u_i^\top(0), \dots, \mathbf u_i^\top(T-1))^\top$ denote the vectors of active power injections,  reactive power injections, and control inputs associated with prosumer $i$ over the entire time horizon, respectively. Denote by $\mathbf p(t)= ( p_1(t), \dots,  p_N(t) )^\top$, $\mathbf q(t)= ( q_1(t), \dots,  q_N(t) )^\top$, and $\mathbf u(t)= (\mathbf u_1^\top(t), \dots, \mathbf u_N^\top(t) )^\top$ the vectors of  active power injections, reactive power injections, and control inputs associated with all prosumers at time step $t \in \mathcal{T}$, respectively. Also, let $\mathbf P=(\mathbf p^\top(0), \dots, \mathbf p^\top(T-1))^\top$, 
  $\mathbf Q=(\mathbf q^\top(0), \dots, \mathbf q^\top(T-1))^\top$, and $\mathbf U=(\mathbf u^\top(0), \dots, \mathbf u^\top(T-1))^\top$ denote  the vectors of all active power injection, reactive power injection, and control inputs associated with all prosumers at all time steps, respectively. 
\subsection{Competitive Market With Locational Pricing}\label{sec:preliminary_locational}
Market design in the presence of grid constraints, such as voltage and thermal constraints, potentially leads to different locational prices at different nodes. For example, \cite{li2015market} considers an islanded grid with a static formulation and proposes a competitive market for energy exchange that respects voltage constraints, leading to different locational prices at different nodes. Inspired by the work in \cite{li2015market}, we proposed a competitive market for energy exchange in  islanded electric grids with voltage constraints   considering a dynamic formulation  \cite{salehi2024acc}. In the following, we describe such a competitive market and the associated equilibrium point. 

Consider the system model introduced in Section \ref{sec:preliminaries}. For simplicity, we neglect thermal constraints and we only focus on voltage constraints in \eqref{eq_voltage}. Also, we suppose that reactive power injection is negligible at each node, i.e., $q_i(t)=0$ for $i \in \mathcal{N}, t \in \mathcal{T}$. Denote by $\lambda_i(t) \in \mathbb{R}$ the locational price at node $i\in \mathcal{N}$ and time step $t \in \mathcal T$ for unit active power exchange.  Prosumers $i \in \mathcal{N}$ are price takers who receive  price signals $\lambda_i(t)$ from the market operator and form a P2P competitive market among themselves. Each prosumer aims to maximize their payoff as the summation of their utilities $f_i(\mathbf x_i(t), \mathbf u_i(t))$  and $\phi_i(\mathbf x_i(T))$, representing their comfort level from consuming energy throughout the entire time horizon and at the terminal time step, and the income or expenditure $\lambda_i(t)p_i(t)$ from energy trading. Such a competitive market reaches an equilibrium if no participant has an incentive to change their decision unilaterally and the total traded power is balanced while the voltage constraints are satisfied.
\subsubsection{Competitive Equilibrium}
 A well-known equilibrium concept for competitive markets  is a \textit{competitive equilibrium} which is defined in the following. 
\begin{definition}[as in \cite{salehi2024acc}]\label{def0}
    Optimal decisions $(\mathbf u_1^\ast(t),  p_1^\ast(t), \dots, \mathbf u_N^\ast(t),  p_N^\ast(t))$, along with optimal locational prices $(\lambda_1^\ast(t), \cdots, \lambda_N^\ast(t))$ for all $t \in \mathcal{T}$, form a \textit{competitive equilibrium} if the following statements hold.

\begin{itemize}
    \item[(i)]  Given $\lambda^\ast_i (t)$, each prosumer $i$ maximizes their payoff  at $(\mathbf u_i^\ast(t),  p_i^\ast(t))$ for $t \in \mathcal{T}$ as a solution to 
\begin{equation}\label{eq_competitive_0}
\begin{aligned}
\max_{{\mathbf U_i}, \mathbf p_i} \quad &  \sum_{t=0}^{T-1} \Big(f_i(\mathbf x_i(t), \mathbf u_i(t)) + \lambda_i^\ast(t)  p_i(t) \Big)  +{\phi_i(\mathbf x_i(T)}) \\
{\rm s.t.} 
\quad &  \mathbf x_i(t+1)= \mathbf A_i \mathbf x_i(t)+ \mathbf B_i \mathbf u_i(t),  \\
\quad &   p_i(t) \leq  a_i (t)-  h_i(\mathbf u_i(t)),  \\
\quad &  \underline{\mathbf x}_i \leq \mathbf x_i(t) \leq \overline{\mathbf x}_i, \,\,\, \underline{\mathbf u}_i\leq \mathbf u_i(t) \leq \overline{\mathbf u}_i,
	\,\,\,  t \in \mathcal{T}.
\end{aligned} 
\end{equation}

\item [(ii)] The total traded power is balanced at each time step, i.e., 
$
\sum_{i=1}^N  p_i^\ast(t) =0
$, for $t \in \mathcal{T}$.

 \item[(iii)] The voltage constraints are satisfied such that
	 $
	 	\underline {v}_i \leq v_0+ \sum_{k =1}^N R_{ik}  p_k^\ast(t) \leq \overline { v}_i
	 $ for $i \in \mathcal{N}$ and $t \in \mathcal{T}$.
	
		\item[(iv)] The locational prices satisfy 
    $\lambda_i^\ast(t)=\alpha^\ast(t)+\sum_{k=1}^{N} \big[\underline\xi_k^\ast(t)-\overline\xi_k^\ast (t) \big] R_{ki}$ for $i \in \mathcal{N}$ and $t \in \mathcal{T}$,
	where $\alpha^\ast(t)$ is the electricity price and $\underline\xi_k^\ast(t)-\overline\xi_k^\ast (t)$ is the  price for a unit voltage change at node $k$.
 
		\item[(v)]  If the voltage constraints are not binding, the price for voltage change is zero for $i \in \mathcal{N}$ and  $t \in \mathcal{T}$, i.e.,
		$
			\underline \xi_i^\ast(t)\Big(\underline v_i - v_0-\sum_{k =1}^N R_{ik}  p_k^\ast(t)\Big)=  \overline \xi_i^\ast(t) \Big( v_0+ \sum_{k =1}^N R_{ik}  p_k^\ast (t)-\overline{ v}_i \Big)=0
		$.
		
\end{itemize}
\end{definition}

\subsubsection{Social Welfare Maximization}
From a system-level perspective, a decision $(\mathbf u_1^\star(t),  p_1^\star(t),  \dots, \mathbf u_N^\star(t),  p_N^\star(t))$ for all $t \in \mathcal{T}$ is efficient if it maximizes the social welfare of the whole society  as a solution to \cite{salehi2024acc}
\begin{equation}\label{eq_social_locational}
		\begin{aligned}
			\max_{{\mathbf U}, \mathbf P} \quad &  \sum_{i=1}^{N}\Big(\sum_{t=0}^{T-1} f_i(\mathbf x_i(t), \mathbf u_i(t))+ {\phi_i(\mathbf x_i(T))} \Big)\\
			{\rm s.t.} 
			\quad &  \mathbf x_i(t+1)= \mathbf A_i \mathbf x_i(t)+ \mathbf B_i \mathbf u_i(t),  \\
			\quad &  p_i(t) \leq a_i (t)- h_i(\mathbf u_i(t)),  \\
			\quad &	\sum_{i=1}^N p_i(t) = 0,
			 \quad  \underline {v}_i \leq v_0+\sum_{k =1}^N R_{ik}  p_k(t) \leq \overline { v}_i, \\
			 \quad &\underline{\mathbf x}_i \leq \mathbf x_i(t) \leq \overline{\mathbf x}_i, \,\,\, \underline{\mathbf u}_i\leq \mathbf u_i(t) \leq \overline{\mathbf u}_i, \,\,\, i \in \mathcal{N}, \,\,\, t \in \mathcal{T}.
		\end{aligned}
	\end{equation}

\begin{assumption}\label{assumption1}
    Suppose $f_i(\cdot)$, $\phi_i(\cdot)$, and $-h_i(\cdot)$ are concave functions for $i \in \mathcal{N}$.
\end{assumption}

\begin{assumption}\label{assumption1_2}
     Slater's condition holds for \eqref{eq_competitive_0} and \eqref{eq_social_locational}. 
\end{assumption}

\begin{theorem}[as in \cite{salehi2024acc}]
    Let Assumptions \ref{assumption1} and \ref{assumption1_2} hold. Given feasible initial conditions $\mathbf x_i(0)$ for $i \in \mathcal{N}$, the competitive equilibrium in Definition \ref{def0} is equivalent to the social welfare maximization solution to \eqref{eq_social_locational}.
\end{theorem}

\subsection{Dynamic Operating Envelopes}
\begin{definition}[as in \cite{mahmoodi2023capacity}]
    A DOE refers to a nonempty, convex, and closed set $\Omega_{it}\subseteq \mathbb{R}^2$ in the $(p_{i}(t), q_i(t))$ coordinates that identifies a range for real and reactive power injections  at node $i$ and time $t$ that can be safely accommodated by the electric grid without violating operational constraints.
\end{definition}
DOEs are determined by a system operator responsible for the safe and secure operation of an electrical grid. To obtain DOEs associated with \eqref{eq:constraints}, we use the right-hand side decomposition (RHSD) method \cite{kornai1965two ,konnov2014right}, and decompose the right-hand side of \eqref{eq:constraints} as
\begin{equation}\label{eq:RHSD}
    \boldsymbol \nu(t) = \sum_{i=1}^N \mathbf w_i(t),
\end{equation}
where $\mathbf w_i(t) \in \mathbb{R}^M$ represents the share of prosumer $i$ in satisfying grid constraints at time interval $t$. Substituting  \eqref{eq:RHSD} into \eqref{eq:constraints} yields
$
    \sum_{i=1}^N \mathbf g_{it}( p_i(t),  q_i(t)) \leq \sum_{i=1}^N \mathbf w_i(t)
$,
which can be decomposed to 
$
   \mathbf g_{it}( p_i(t),  q_i(t)) \leq  \mathbf w_i(t)
$, for $i \in \mathcal{N}$.
Then, DOEs can be obtained as 
$
    \Omega_{it}=\big\{(p_{i}(t), q_i(t)) \in \mathbb{R}^2: \mathbf g_{it}( p_i(t),  q_i(t)) \leq \mathbf w_i(t)\big\}$
for each prosumer $i\in \mathcal{N}$ at time step $t\in \mathcal{T}$ \cite{mahmoodi2023capacity}.

There exists an infinite number of configurations for $\mathbf w_i(t)$ that satisfy \eqref{eq:RHSD}, and therefore an infinite number of DOEs, among which the system operator looks for the one that optimizes its objective $\mathbb F (\mathbf p(t), \mathbf q(t))$. System operators can have different objectives in determining DOEs, such as maximizing social welfare, minimizing loss, and maximizing grid capacity. In this paper, we suppose that the operator aims to maximize the capacity of the electric grid by maximizing the total possible active power injection to the grid at each time step $t \in \mathcal{T}$, i.e., $\mathbb F (\mathbf p(t), \mathbf q(t))= \sum_{i=1}^N \max \{0, p_i(t)\}$. Denote $\mathbf w(t)=(\mathbf w_1^\top(t), \dots, \mathbf w_N^\top(t))^\top$. An optimal  $\mathbf w(t)$ that serves the objective of the system operator is then obtained solving
\begin{equation}\label{eq:DOE_2}
		\begin{aligned}
			\max_{\mathbf w(t), \mathbf p(t), \mathbf q(t)} \quad &  \sum_{i=1}^{N} \max\{0, p_i(t)\} - \epsilon O(\mathbf w(t)) \\
			{\rm s.t.} 
			\quad &
			 \mathbf g_{it}( p_i(t),  q_i(t)) \leq \mathbf w_i(t) \\
    \quad & \sum_{i=1}^N \mathbf w_i(t) = \boldsymbol \nu(t),
    \,\,\, i \in \mathcal{N},
		\end{aligned}
	\end{equation}
where $\epsilon O(\mathbf w(t))$ is a regularization term and $\epsilon$ is a small positive number. The function $O(\mathbf w(t)): \mathbb{R}^{NM} \mapsto \mathbb{R}$ indicates the  preferences of the operator in assigning different shares of DOEs to prosumers. For example, the operator might prefer to assign larger DOEs to prosumers closer to the feeder, or it might prefer to assign near equal DOEs to all prosumers. 
Following \cite{mahmoodi2023capacity}, we adopt the notion of close-to-equal contributions to promote fairness in assigning DOEs.
Denote by  $\overline{\mathbf{E}} \in \mathbb R^{NM}, \overline{\mathbf{E}}=(\boldsymbol \nu^\top(t)/N, \dots, \boldsymbol \nu^\top(t)/N)^\top$ the equality index. A function $O(\mathbf w(t))$ that leads to close-to-equal contributions can be defined as 
    \begin{equation}    \label{eq:close_to_equal}  
    O(\mathbf w(t))=\|\mathbf w(t)- \overline{\mathbf{E}} \|.
    \end{equation}
    
Then, solving \eqref{eq:DOE_2}, the operator obtains $\mathbf w^\ast(t)=(\mathbf w_1^{\ast\top}(t), \dots, \mathbf w_N^{\ast\top}(t))^\top$ and sends the DOEs 
\begin{equation}\label{eq_DOE_3}
\mathbf g_{it}( p_i(t),  q_i(t)) \leq \mathbf w_i^\ast(t)   
\end{equation}
to the associated prosumers $i \in \mathcal{N}$ at each time step $t \in \mathcal{T}$. 

%
\section{Competitive Market With Uniform Pricing}\label{sec:competitive}
Market design in the presence of grid constraints \eqref{eq:constraints} potentially leads to different locational prices at different nodes. 
Such locational price differences cause social inequities that affect market participants differently based on their geographic location in the grid \cite{tangeras2021competitive, edmunds2017distributed, CitiPower}. To address this issue, we propose two types of competitive markets in islanded radial grids that ensure uniform pricing for all participants while satisfying grid constraints. We use the concept of DOEs  to decompose grid constraints \eqref{eq:constraints} into local constraints \eqref{eq_DOE_3} that can be accommodated by each prosumer independently. 

In what follows, for simplicity, we suppose that the reactive power injection at each node is negligible, i.e., $q_i(t)=0$ for $i \in \mathcal{N}$, $t \in \mathcal{T}$. Extensions to consider the regulation of reactive power injection at each node are straightforward.


\subsection{Competitive Market With DOEs}\label{sec:uniform_1}
A P2P competitive energy market for the islanded radial grid introduced in Section \ref{sec:preliminaries} can be designed in which prosumers are price takers and they receive a uniform  price signal $\lambda(t) \in \mathbb{R}$ for unit energy trading from the market operator at each time step $t \in \mathcal{T}$. Prosumers trade their local electricity supply with their peers according to the announced price signal and the assigned DOEs $\mathbf g_{it}( p_i(t),  q_i(t)) \leq \mathbf w_i^\ast(t)$. They aim to maximize their payoff as the summation of their utility $f_i(\mathbf x_i(t), \mathbf u_i(t))$ and $\phi_i(\mathbf x_i(T))$, representing their comfort/satisfaction level from energy consumption throughout the entire time horizon and at the terminal time step, respectively, and the income or expenditure $\lambda(t)p_i(t)$ from energy transaction. Since energy transaction is done through the electric grid, the grid constraints must be always satisfied to enable a safe and secure system operation. Using DOEs, the grid constraints are decomposed to some local constraints that would be accommodated by each participant independently. Such a market is at an equilibrium if no participant has an incentive to change their decision and the  total traded power is balanced.  

\subsubsection{Competitive Equilibrium}
An equilibrium point for the prescribed competitive market falls into the category of \textit{competitive equilibrium} and is defined in the following.
\begin{definition}\label{def2}
   Optimal decisions $(\mathbf u_1^\ast(t),  p_1^\ast(t),  \dots, \mathbf u_N^\ast(t),  p_N^\ast(t))$, along with uniform optimal prices $ \lambda^\ast(t)$ for all $t \in \mathcal{T}$, form a \textit{competitive equilibrium} if the following statements hold.
   \begin{itemize}
       \item [(i)] Given $ \lambda^\ast(t)$ and $\mathbf w_i^\ast(t)$, each prosumer $i \in \mathcal{N}$ maximizes their payoff at $(\mathbf u_i^\ast(t), p_i^\ast(t))$ for $t \in \mathcal{T}$ as a solution to
       \begin{equation} \label{eq_competitive_2}
			\begin{aligned}		\max_{{\mathbf U_i}, \mathbf p_i} \quad &  \sum_{t=0}^{T-1} \Big(f_i(\mathbf x_i(t), \mathbf u_i(t)) +\lambda^\ast(t) p_i(t) \Big)+{\phi_i(\mathbf x_i(T)}) \\
				{\rm s.t.} 
				\quad &  \mathbf x_i(t+1)= \mathbf A_i \mathbf x_i(t)+ \mathbf B_i \mathbf u_i(t),  \\
				\quad &  p_i(t) \leq a_i (t)- h_i(\mathbf u_i(t)),  \,\,\, \mathbf g_{it}( p_i(t)) \leq \mathbf w_i^\ast(t), 
    \\
				\quad &
    \underline{\mathbf x}_i \leq \mathbf x_i(t) \leq \overline{\mathbf x}_i, 
   \,\,\, \underline{\mathbf u}_i\leq \mathbf u_i(t) \leq \overline{\mathbf u}_i, 
				\,\,\,  t \in \mathcal{T}.
			\end{aligned}
		\end{equation}

  \item [(ii)] All traded power is balanced at each time step; that is,
  $
      \sum_{i=1}^N p_i^\ast(t)=0
  $ for $t \in \mathcal{T}$.
   \end{itemize}
\end{definition}

\subsubsection{Social Welfare Maximization}
From a system-level perspective, a decision $(\mathbf u_1^\star(t),  p_1^\star(t),  \dots, \mathbf u_N^\star(t),  p_N^\star(t))$ for all $t \in \mathcal{T}$ is efficient if it maximizes the social welfare of the whole society, defined  as the summation of all utilities $f_i(\cdot)$ and $\phi_i(\cdot)$ of all participants, as a solution to
\begin{equation}\label{eq_social_2}
		\begin{aligned}
			\max_{{\mathbf U}, \mathbf P} \quad &  \sum_{i=1}^{N}\Big(\sum_{t=0}^{T-1} f_i(\mathbf x_i(t), \mathbf u_i(t))+ {\phi_i(\mathbf x_i(T))} \Big)\\
			{\rm s.t.} 
			\quad &  \mathbf x_i(t+1)= \mathbf A_i \mathbf x_i(t)+ \mathbf B_i \mathbf u_i(t),  \\
			\quad &  p_i(t) \leq a_i (t)- h_i(\mathbf u_i(t)),   \\
				\quad & \mathbf g_{it}( p_i(t)) \leq \mathbf w_i^\ast(t), \,\,\,	\sum_{i=1}^N p_i(t) = 0,     \\
				\quad &
    \underline{\mathbf x}_i \leq \mathbf x_i(t) \leq \overline{\mathbf x}_i, 
   \,\,\, \underline{\mathbf u}_i\leq \mathbf u_i(t) \leq \overline{\mathbf u}_i, 
				\,\,\,   i \in \mathcal{N}, \,\,\, t \in \mathcal{T}.
		\end{aligned}
	\end{equation}

\begin{assumption}\label{assumption2}
 Slater's condition holds for \eqref{eq_competitive_2} and \eqref{eq_social_2}. 
\end{assumption}

\begin{theorem}\label{theorem2}
    Let Assumptions \ref{assumption1} and \ref{assumption2} hold. Given feasible initial conditions $\mathbf x_i(0)$ for $i \in \mathcal{N}$, the following statements hold.
    \begin{itemize}
        \item A competitive equilibrium in Definition \ref{def2} is equivalent to a social welfare maximization solution to \eqref{eq_social_2}. 
        \item Let $-\alpha^\ast(t)$ be an optimal dual variable corresponding to the balancing equality constraint $\sum_{i=1}^N p_i(t)=0$ in \eqref{eq_social_2} for $t\in \mathcal{T}$. Then, a uniform equilibrium price can be obtained as $\lambda^\ast(t)=\alpha^\ast(t)$.
        \end{itemize}
\end{theorem}
\begin{proof}
    See Appendix \ref{Appendix_Theorem2}.  
\end{proof}

Using DOEs to decompose the global constraints into local ones makes the approach conservative. That is, in some cases, the assigned DOEs are potentially infeasible to implement in the proposed competitive market due to insufficient generation or consumption. 
To mitigate this conservatism, we introduce an additional degree of freedom in decision-making by designing a competitive market that permits limit trading in DOEs. 

\subsection{Competitive Market With Limit Trading in DOEs}\label{sec:main_limit}
The idea of limit trading in DOEs was first introduced in \cite{umer2023novel}, which enhances the utilization of grid capacity in the presence of DOEs. If some prosumers  have excess DOE limits $\mathbf w_i^\ast(t)-\mathbf g_{it}(p_i(t))$, they can sell them to those who have shortage of DOE limits. Denote by $\boldsymbol l_i(t) \in \mathbb R^M$ the traded limit of prosumer $i$ that is physically constrained by $\boldsymbol l_i(t) \leq \mathbf w_i^\ast(t)-\mathbf g_{it}(p_i(t))$.
Limit trading introduces an additional degree of freedom for prosumers to trade their excess of  DOE limits, which are not otherwise used, and gain profit. Denote by $\boldsymbol \beta(t)\in \mathbb{R}^M$ the uniform limit price for unit traded DOE limit. Obtaining the DOE limits  $\mathbf{w}_i^\ast(t)$ according to \eqref{eq:DOE_2},
 the coordinator identifies optimal uniform electricity prices $\lambda^\ast(t)$ and limit prices $\boldsymbol \beta^\ast(t)$ that clear the market for $t \in \mathcal{T}$, which are then broadcast to market participants. After receiving the uniform electricity prices and limit prices, prosumers form a P2P competitive market to trade their energy and limit resources and gain profit. Each prosumer aims to maximize their payoff as the summation of their utility $f_i(\mathbf x_i(t), \mathbf u_i(t))$ and $\phi_i(\mathbf x_i(T))$, and the income $\lambda(t)p_i(t)+\boldsymbol\beta(t) \cdot \boldsymbol l_i(t)$ from electricity and limit trading over the entire time horizon. 
 
 \subsubsection{Competitive Equilibrium}
 An equilibrium point for such a market falls into the category of \textit{competitive equilibrium} and is defined in the following.
 \begin{definition}\label{def3}
     Optimal decisions $(\mathbf u_1^\ast(t),  p_1^\ast(t), \boldsymbol l_1^\ast(t), \dots, \mathbf u_N^\ast(t),  p_N^\ast(t), \boldsymbol l_N^\ast(t))$, along with optimal uniform prices $ \lambda^\ast(t)$ and $\boldsymbol \beta^\ast(t)$ for all $t \in \mathcal{T}$, form a \textit{competitive equilibrium} if the following statements hold.
     \begin{itemize}
         \item [(i)] Given $\lambda^\ast(t)$, $\boldsymbol \beta^\ast(t)$, and $\mathbf w_i^\ast(t)$, each prosumer $i \in \mathcal{N}$ maximizes their payoff at $(\mathbf u_i^\ast(t), p_i^\ast(t), \boldsymbol l_i^\ast(t))$ for $t \in \mathcal{T}$ as a solution to       
         \begin{equation}\label{eq17}
			\begin{aligned}
				\max_{{\mathbf U_i}, \mathbf p_i, \mathbf L_i} \quad &  \sum_{t=0}^{T-1} f_i(\mathbf x_i(t), \mathbf u_i(t)) + {\phi_i(\mathbf x_i(T)}) \\ &+ \sum_{t=0}^{T-1} \Big( \lambda^\ast(t) p_i(t)+\boldsymbol  \beta^\ast(t) \cdot \boldsymbol l_i(t)\Big) \\
				{\rm s.t.} 
				\quad &  \mathbf x_i(t+1)= \mathbf A_i \mathbf x_i(t)+ \mathbf B_i \mathbf u_i(t),  \\
				\quad &  p_i(t) \leq a_i (t)- h_i(\mathbf u_i(t)),  \\
				\quad & \boldsymbol l_i(t) \leq \mathbf w_i^\ast(t)-\mathbf g_{it}( p_i(t)), \\
				\quad &
    \underline{\mathbf x}_i \leq \mathbf x_i(t) \leq \overline{\mathbf x}_i, 
   \,\,\, \underline{\mathbf u}_i\leq \mathbf u_i(t) \leq \overline{\mathbf u}_i, 
				\,\,\,  t \in \mathcal{T},
			\end{aligned}
		\end{equation}
where $\mathbf L_i=(\boldsymbol l_i^\top(0), \dots, \boldsymbol l_i^\top(T-1))^\top$ is the vector of all traded limits associated with participant $i$ over the entire time horizon.
  \item [(ii)] All traded power and all traded limit are balanced at each time step; that is,
  $
      \sum_{i=1}^N p_i^\ast(t)=0 $ and $\sum_{i=1}^N \boldsymbol l_i^\ast(t)=0$ for $t\in \mathcal{T}
  $.
     \end{itemize}
 \end{definition}

 \subsubsection{Social Welfare Maximization}
 From a system-level perspective, a decision $(\mathbf u_1^\star(t),  p_1^\star(t), \boldsymbol l_1^\star,  \dots, \mathbf u_N^\star(t),  p_N^\star(t), \boldsymbol l_N^\star)$ for all $t \in \mathcal{T}$ is efficient if it maximizes the social welfare of the whole society, defined  as the summation of all utilities $f_i(\cdot)$ and $\phi_i(\cdot)$ of all participants, as a solution to
 \begin{equation}\label{eq2}
		\begin{aligned}
			\max_{{\mathbf U}, \mathbf P, \mathbf L} \quad &  \sum_{i=1}^{N}\Big(\sum_{t=0}^{T-1} f_i(\mathbf x_i(t), \mathbf u_i(t))+ {\phi_i(\mathbf x_i(T))} \Big)\\
			{\rm s.t.} 
			\quad &  \mathbf x_i(t+1)= \mathbf A_i \mathbf x_i(t)+ \mathbf B_i \mathbf u_i(t),  \\
			\quad &  p_i(t) \leq a_i (t)- h_i(\mathbf u_i(t)),   \\
				\quad & \boldsymbol l_i(t) \leq \mathbf w_i^\ast(t)-\mathbf g_{it}( p_i(t)) \\
			\quad &	\sum_{i=1}^N p_i(t) = 0, \,\,\,	\sum_{i=1}^N \boldsymbol l_i(t) = 0, 
    \\
				\quad &
    \underline{\mathbf x}_i \leq \mathbf x_i(t) \leq \overline{\mathbf x}_i, 
   \,\,\, \underline{\mathbf u}_i\leq \mathbf u_i(t) \leq \overline{\mathbf u}_i, 
				\,\,\,   i \in \mathcal{N}, \,\,\, t \in \mathcal{T},
		\end{aligned}
	\end{equation}
where $\mathbf L=(\mathbf L_1^\top, \dots, \mathbf L_N^\top)^\top$ is the vector of all traded limits associated with all participants over the entire time horizon.

\begin{assumption}\label{assumption3}
 Slater's condition holds for \eqref{eq17} and \eqref{eq2}. 
\end{assumption}
\begin{theorem}\label{theorem3}
Let Assumptions \ref{assumption1} and \ref{assumption3} hold. Given feasible initial conditions $\mathbf x_i(0)$ for $i \in \mathcal{N}$, the following statements hold.
\begin{itemize}
    \item A competitive equilibrium in Definition \ref{def3} is equivalent to a social welfare maximization solution to \eqref{eq2}. 
    \item Let $-\alpha^\ast(t)$ and $-\boldsymbol \delta^\ast(t)$ be optimal dual variables corresponding to the balancing equality constraints $\sum_{i=1}^N p_i(t)=0$ and $\sum_{i=1}^N \boldsymbol l_i(t)=0$ in \eqref{eq2}, respectively, for $t\in \mathcal{T}$. Then,  uniform equilibrium prices can be obtained as $\lambda^\ast(t)=\alpha^\ast(t)$ and $\boldsymbol \beta^\ast(t)=\boldsymbol \delta^\ast(t)$ for $t\in \mathcal{T}$.
    \end{itemize}
\end{theorem}
\begin{proof}
    See Appendix \ref{Appendix_Theorem3}.
\end{proof}

\subsection{Connection Between Locational and Uniform Pricing}
In this section, we examine the connection between the competitive market with limit trading in DOEs, proposed in Section \ref{sec:main_limit}, and the competitive market with locational pricing introduced in Section \ref{sec:preliminary_locational}. For simplicity, we consider only voltage constraints in the grid to be consistent with the formulation provided in Section \ref{sec:preliminary_locational}.

\begin{theorem}\label{theorem4}
    Let Assumptions \ref{assumption1}, \ref{assumption1_2}, and \ref{assumption3} hold. Then, the competitive market with limit trading in DOEs, proposed in Section \ref{sec:main_limit}, is equivalent to the competitive market with locational pricing introduced in Section \ref{sec:preliminary_locational}. To be precise, if $(\mathbf U^\star, \mathbf P^\star)$ maximizes \eqref{eq_social_locational}, then there exists $\mathbf L^\star$ such that $(\mathbf U^\star, \mathbf P^\star, \mathbf L^\star)$ maximizes \eqref{eq2}, and vice versa.
\end{theorem}
\begin{proof}
   See Appendix \ref{Appendix_Theorem4}. 
\end{proof}

The next theorem shows that  optimal uniform prices proposed in Section \ref{sec:main_limit} are part of  optimal locational prices introduced in Section \ref{sec:preliminary_locational}. From Definition \ref{def0}, the locational prices are given by $\lambda_i^\ast(t)=\alpha^\ast(t)+\sum_{k=1}^{N} \big[\underline\xi_k^\ast(t)-\overline\xi_k^\ast (t) \big] R_{ki}$ for $i \in \mathcal{N}$ and $ t \in \mathcal{T}$. For simplicity, define $\alpha^\ast_i(t)=\sum_{k=1}^{N} \big[\underline\xi_k^\ast(t)-\overline\xi_k^\ast (t) \big] R_{ki}$, which represents the second part of the locational prices. Thus, $\lambda_i^\ast(t)=\alpha^\ast(t)+\alpha^\ast_i(t)$.

\begin{theorem}\label{theorem5}
     Let Assumptions \ref{assumption1}, \ref{assumption1_2}, and \ref{assumption3} hold. Consider optimal locational prices in Section \ref{sec:preliminary_locational} which are defined as $\lambda_i^\ast(t)=\alpha^\ast(t)+\alpha^\ast_i(t)$. Then, optimal uniform  prices $\lambda^\ast(t)$ obtained in the competitive market approach with limit trading in DOEs, proposed in Section \ref{sec:main_limit}, satisfy  $\lambda^\ast(t)=\alpha^\ast(t)$.
\end{theorem}
\begin{proof}
See Appendix \ref{Appendix_Theorem5}.
\end{proof}

\begin{definition}[as in \cite{li2020transactive, ma2020incentive}]
A mechanism is strongly budget balanced (or weakly budget balanced) if the summation of all payments collected by the system coordinator is zero (or nonegative).
\end{definition}

\begin{proposition}\label{proposition_2}
The following statements hold. 
\begin{itemize}
    \item[(i)] The competitive market with locational pricing in Section \ref{sec:preliminary_locational} is weakly budget balanced under the competitive equilibria, i.e.,  $-\sum_{i=1}^N \lambda_i^\ast(t)p_i^\ast(t) \geq 0$ for $t \in \mathcal{T}$.
    \item[(ii)] In contrast, the proposed uniform pricing approaches guarantee that the market is strongly budget balanced under the competitive equilibria, i.e., $\sum_{i=1}^N \lambda^\ast(t)p_i^\ast(t)=0$ and $\sum_{i=1}^N \boldsymbol  \beta^\ast(t) \cdot \boldsymbol l_i^\ast(t)=0$ for $t \in \mathcal{T}$.
\end{itemize}
\end{proposition}

\begin{proof}
See Appendix \ref{Appendix_Proposition_2}.
\end{proof}

Consider the competitive market with locational pricing in Section \ref{sec:preliminary_locational}. If the market operator collects all payments, which are not zero in general, and distributes the associated budget surplus $-\sum_{i=1}^N \lambda_i^\ast(t)p_i^\ast(t)$ to prosumers equally, it is equivalent to  assigning the same DOE limits $\mathbf w_i^\ast(t)=\boldsymbol \nu (t) /N$ to all prosumers in  the design of the competitive market with limit trading in DOEs introduced in Section \ref{sec:main_limit}. In other words, if the system operator assigns different DOE limits $\mathbf w_i^\ast(t)$ to prosumers in  Section \ref{sec:main_limit}, it is equivalent to distributing the budget surplus $-\sum_{i=1}^N \lambda_i^\ast(t)p_i^\ast(t)$, obtained from the locational pricing approach, to prosumers in a nonuniform manner. This result is proved in the following proposition.

\begin{proposition}\label{proposition1}
    Consider the competitive market with locational pricing in Section \ref{sec:preliminary_locational}. Distributing the budget surplus $-\sum_{i=1}^N \lambda_i^\ast(t)p_i^\ast(t)$ equally to  prosumers is equivalent to assigning prosumers the same DOE limits $\mathbf w_i^\ast(t)=\boldsymbol \nu (t) /N$ in  Section \ref{sec:main_limit}, which considers uniform pricing with limit trading in DOEs. 
\end{proposition}
\begin{proof}
    See Appendix \ref{Appendix_Proposition1}.
\end{proof}


\section{Numerical Results}\label{sec:Numerical Results}
In this section, we examine the proposed approaches on a modified  IEEE 13-node test feeder as depicted in Fig. \ref{fig_feeder}. We consider a single-phase representation that omits the transformer between nodes 2 and 3, the switch between nodes 6 and 7, and all capacitor banks \cite{nimalsiri2021coordinated}. All simulations are done in Python 3.8 + CVXPY on a MacBook Pro with an Apple M1 chip and 16-GB memory.

\begin{figure}[!t]
	\centering
	\includegraphics[width=3.35 in]{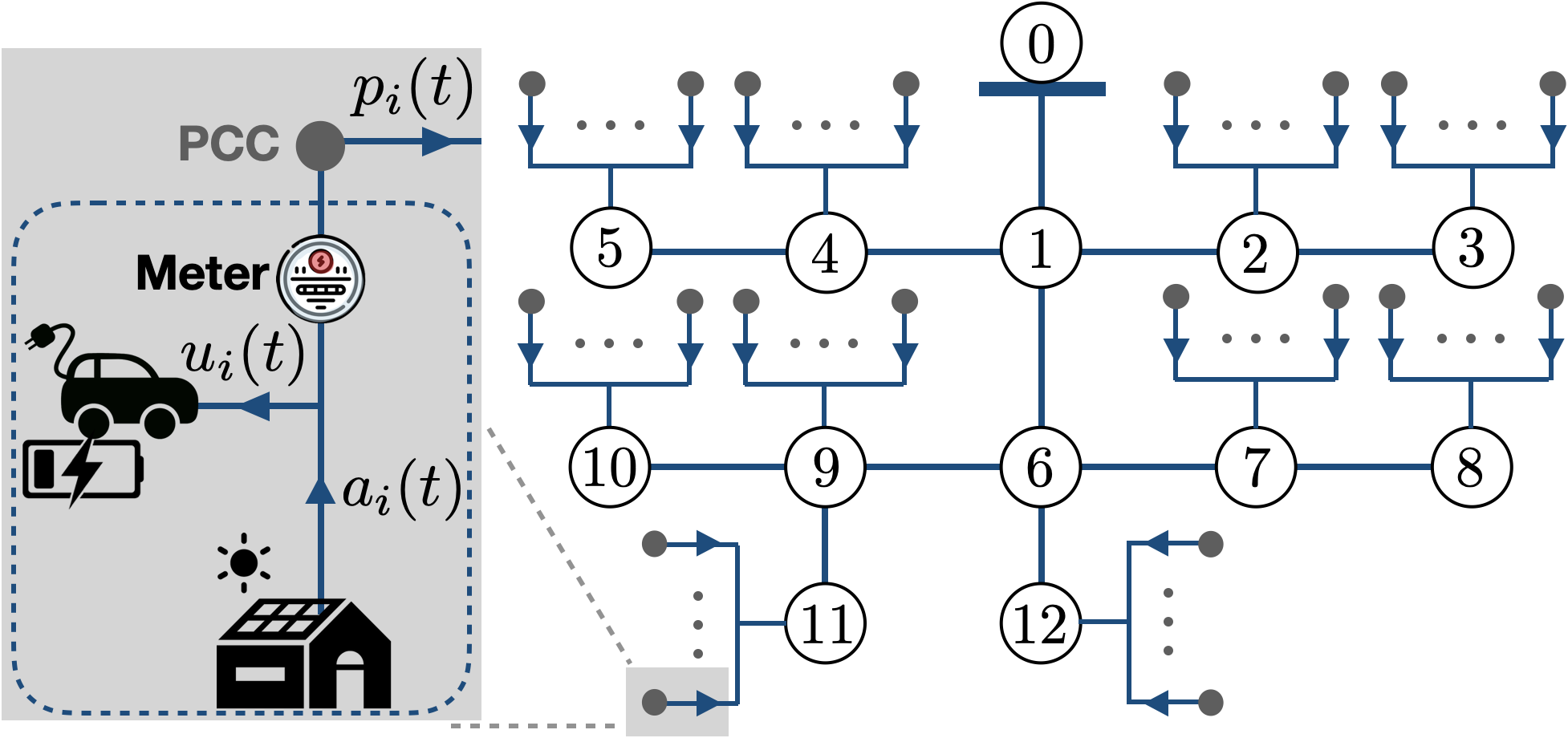}
	\caption{Modified IEEE 13-node test feeder. Left: A prosumer $i$ is zoomed in, connected to the electric grid at the point of common coupling (PCC).}
	\label{fig_feeder}
\end{figure}

\subsection{Simulation Data}
Each node---except for the feeder node, node $1$, and node $6$---represents a combination of $30$ distribution aggregators. Each aggregator manages $12$ residential prosumers equipped with gridable EVs and residential battery storage systems. The net supply of each prosumer, defined as the difference between their solar PV generation and uncontrollable loads, is based on metering data collected by Ausgrid \cite{Ausgrid1}, an electricity distribution company in Australia. This data is scaled by a factor of $2$ to account for the increase in solar PV generation and electrification compared to the year 2012, when the data was collected.

Prosumers have various EV models, as detailed in \cite{Battery1}, with battery capacities ranging from $0$ to $75$ (kWh). 
They also have home battery storage systems with capacities in the same range. Both EVs and home battery storage systems are assumed to have charge/discharge efficiencies of $\eta=0.9$. We study the system over a $24$ hour (h) time period starting from midnight with a sampling time of $\Delta=0.5$ (h). The dynamics of each aggregator $i$, which accounts for both EVs and home battery storage systems of prosumers, is modeled using the linear difference equation $\mathbf x_i(t+1)=\mathbf x_i(t)+\eta \mathbf u_i(t)\Delta$, where $\mathbf x_i(t) \in \mathbb R^2$ represents the SoC vector of the aggregated EVs and residential battery storage systems (kWh), and $\mathbf u_i(t) \in \mathbb R^2$ denotes the aggregated charge/discharge rate vector at each time step $t$. Note that, in all two-dimensional vectors, the first element corresponds to the aggregated EVs, while the second element corresponds to the aggregated residential batteries associated with the prosumers of aggregator $i$. The initial SoC of  batteries is assumed to be uniformly distributed in $[0.2, 0.5]$ times the associated battery capacities. 

While the home battery storage systems are always connected to the grid, EVs arrival and departure times are selected based on survey data collected by the Victorian Department of Transport in Australia \cite{VISTA1, Mendeley1}. It is assumed that  prosumers are equipped with  level-$2$ chargers with a maximum charge rate of $6.6$ (kW) \cite{nimalsiri2021coordinated}. In addition, the maximum discharge rate is assumed to be $-6.6$ (kW). To extend battery lifespan, the SoC of EV batteries must remain between $20\%$ and $85\%$ of their capacity between arrival and departure times. Similarly, home battery storage systems must maintain their SoC within the same range. Thus, the aggregated SoC and charge/discharge rates associated with each aggregator $i$ are constrained by $0.2 \mathbf C_i\leq \mathbf x_i(t)\leq 0.85 \mathbf C_i$  and $-6.6 \mathbb 1 \leq \mathbf u_i(t) \leq 6.6 \mathbb 1$, respectively, where  $\mathbf C_i \in \mathbb R^2$ is the vector of aggregated battery capacities (kWh) and $\mathbb 1 \in \mathbb R^2$ denotes a vector with all entries equal to $1$. The first element of $\mathbf C_i$ corresponds to the aggregated EVs, while the second element corresponds to the aggregated residential batteries associated with the prosumers of aggregator $i$. 

\subsection{Experiments}
In the prescribed setting, the $N=300$ aggregators are decision makers who model the average behavior of their associated prosumers. For each aggregator $i \in \mathcal{N}$, the utility functions are quadratic such that $f_i(\mathbf x_i(t), \mathbf u_i(t))=- \vartheta_1\| \mathbf u_i(t)\|^2- \vartheta_{2,t} \|(\mathbf x_i(t)-0.85 \mathbf C_i) \cdot (1,0)^\top\|^2 - \vartheta_{3,t} \|(\mathbf x_i(t)-0.85 \mathbf C_i) \cdot (0,1)^\top\|^2$ for $t \in \mathcal{T}=\{0, \dots, 47\}$, where $\|\cdot\|$ is the Euclidean norm. The terminal utility is assigned to be $\phi_i(\mathbf x_i(T))=-\vartheta_4\|\mathbf x_i(T)-0.85 \mathbf C_i\|^2$,  where  $T=48$ is the terminal time step. The coefficients  $\vartheta_1$ ($10^{-8}$\textcent/(kW)$^2$), $\vartheta_{2,t}$ ($10^{-8}$\textcent/(kWh)$^2$), $\vartheta_{3,t}$ ($10^{-8}$\textcent/(kWh)$^2$), and $\vartheta_4$ ($10^{-8}$\textcent/(kWh)$^2$) are the regularization terms. In the context of EV and battery charging, the consumed or withdrawn power by the $i$-th aggregated controllable load is $h_i(\mathbf u_i(t))=\mathbf u_i(t) \cdot \mathbb 1$. The reference voltage at the feeder is $V_0=1$ (p.u.). According to ANSI C84.1 standard, the voltage magnitude at each node must remain within $\pm5 \% $ of the nominal voltage. 
In other words,  the lower and upper bounds on the squared voltage magnitude are $\underline {v}_i=0.95^2$ (p.u.)$^2$ and $\overline {v}_i=1.05^2$ (p.u.)$^2$, respectively. 
\subsubsection{Uniform Pricing}
 We consider voltage constraints in the grid and obtain DOEs according to \eqref{eq:DOE_2}, where $O(\mathbf w(t))$ represents close-to-equal contributions as in \eqref{eq:close_to_equal}. 
 We suppose the reactive power injection is negligible, i.e., $q_i(t)=0$ for $i \in \mathcal{N}, t \in \mathcal{T}$. We apply the competitive equilibrium in Definition \ref{def3} that allows limit trading in DOEs. Solving \eqref{eq2}, the coordinator obtains uniform energy prices $\lambda^\ast(t)$ (in $10^{-8}$\textcent/kW) and limit trading prices $\boldsymbol\beta^\ast(t)$ (in $10^{-8}$\textcent/(kV)$^2$) as the optimal dual variables corresponding to the balancing equality constraints $\sum_{i=1}^N p_i(t) = 0$ and $\sum_{i=1}^N \boldsymbol l_i(t) = 0$, respectively, which are then broadcast to aggregators. Aggregators then decide about the EV and home battery charging profiles of their associated prosumers as well as their traded energy and DOE limit to maximize their payoff in \eqref{eq17}.  Limit prices $\boldsymbol \beta^\ast(t)$ are vectors of dimension 24, where each element represents either the lower or upper bound constraint on the voltage at one of the 12 nodes in the grid, excluding the feeder node. We denote the $k$-th element of  $\boldsymbol \beta^\ast(t)$ at time step $t$ as $\beta^\ast_{k}(t)$, where $k \in \{1, \dots, 24\}$. 

In Figs. \ref{fig_price_uniform_1} and \ref{fig_price_uniform_2}, we present uniform prices $\lambda^\ast(t)$ and $\boldsymbol\beta^\ast(t)$, respectively. For clarity, the price units have been converted to \textcent/kWh and \textcent/(kV)$^2$  in these figures. Aggregators at the same node share the same voltage. We denote the voltage at node $j$ and time step $t$ as $\vee_j(t)$, where $j \in \{1, \dots, 12\}$. For illustration, Fig. \ref{fig_voltage_1} shows the nodal voltages at all 12 nodes. The horizental dashed lines represent the lower and upper bounds on the voltage. As can be seen, we have uniform prices in the grid although nodal voltages hit the boundaries at some time steps. The net payment collected by the grid coordinator is observed to be zero.
\begin{figure}[!t]
	\centering
	\includegraphics[width=3 in]{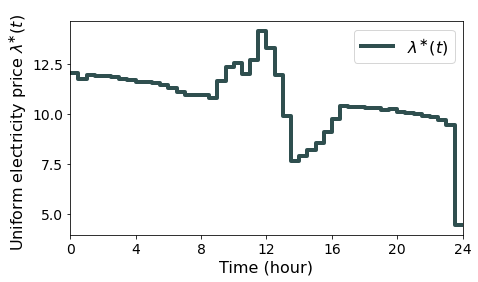}
	\caption{Uniform prices $\lambda^\ast(t)$ (\textcent/kWh) for unit electricity exchange over a 24-hour period (48 time steps of length $0.5$ hour) in the uniform pricing approach with DOE limit trading. All prosumers observe a uniform $\lambda^\ast(t)$.}
	\label{fig_price_uniform_1}
\end{figure}

\begin{figure}[!t]
	\centering
	\includegraphics[width=3 in]{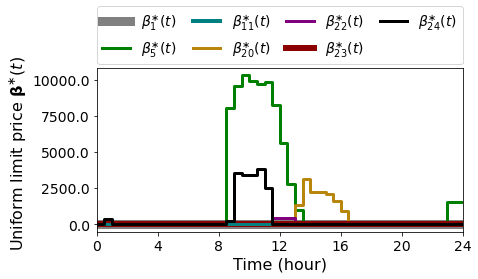}
	\caption{Uniform limit prices  $\boldsymbol\beta^\ast(t)$ (in \textcent/(kV)$^2$) for unit DOE limit exchange over a 24-hour period (48 time steps of length $0.5$ hour).  All prosumers observe a uniform vector of limit prices $\boldsymbol\beta^\ast(t)$. The $k$-th element of  vector $\boldsymbol \beta^\ast(t)$ is denoted by $\beta^\ast_{k}(t)$. The other $17$ elements of the vector $\boldsymbol \beta^\ast(t)$ are close to zero, as their associated element of DOE limit is in surplus, and not depicted in the figure.}
	\label{fig_price_uniform_2}
\end{figure}

\begin{figure}[!t]
	\centering
	\includegraphics[width=3 in]{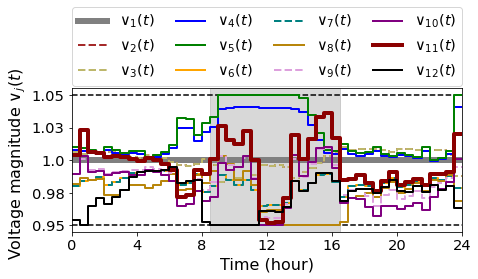}
	\caption{Nodal voltages $\vee_j(t)$ (p.u.) at each node $j$ over  a 24-hour period (48 time steps of length $0.5$ hour) in both the uniform pricing approach with DOE limit trading  and the locational pricing approach. The horizontal dashed lines represent the lower and upper bounds on the voltage.}
	\label{fig_voltage_1}
\end{figure}

\begin{remark}
    This example is infeasible to be solved using the uniform pricing approach without limit trading in DOEs, as introduced in Section \ref{sec:uniform_1}. This confirms that allowing limit trading in DOEs helps mitigate the conservatism associated with the approach in Section \ref{sec:uniform_1}.
\end{remark}

\subsubsection{Locational Pricing}
For the purpose of comparison, we apply the competitive equilibrium with locational pricing in Definition \ref{def0}. The grid coordinator solves the social welfare maximization problem in \eqref{eq_social_locational} and obtaines locational prices $\lambda^\ast_i(t)$ as the combination of the optimal dual variables corresponding to the power balance constraint and voltage constraints \cite{salehi2024acc}, which is then sent to the associated aggregators. After receiving the locational prices, aggregators decide about the EV and home battery charging profiles of their prosumers as well as their traded energy to maximize their payoff in \eqref{eq_competitive_0}. The locational prices for aggregators within the same node are identical. We denote the locational price at node $j$ and time $t$ as $\Lambda^\ast_j(t)$, where $j \in \{1, \dots, 12\}$. For illustration, Fig. \ref{fig_price_locational_1} presents the locational prices at all 12 nodes. The nodal voltages are observed to be the same as the uniform pricing approach, as depicted in Fig. \ref{fig_voltage_1}. Figure \ref{fig_price_locational_1} illustrates that nodal prices differ when  nodal voltages reach their boundaries, e.g., during time steps $16 < t < 34$  which is indicated by the gray shaded area in Figs. \ref{fig_voltage_1} and \ref{fig_price_locational_1}. The budget surplus collected by the grid coordinator is \$637.

\begin{figure}[!t]
	\centering
	\includegraphics[width=3.35 in]{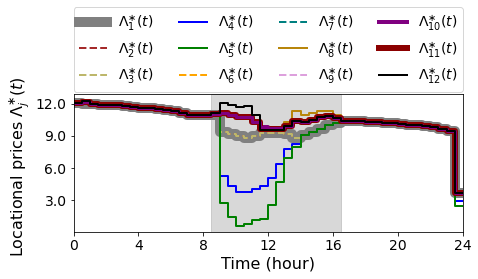}
	\caption{Locational prices $\Lambda^\ast_j(t)$ (\textcent/kWh) for unit electricity exchange at each node $j$ over a 24-hour period (48 time steps of length $0.5$ hour). Locational prices are different when the voltage constraints are binding.}
	\label{fig_price_locational_1}
\end{figure}


\subsubsection{Discussion}
Figure \ref{fig_income_1} depicts the net income/expenditure of the 300 aggregators. As shown, all aggregators earn more under the uniform pricing approach with DOE limit trading compared to the locational pricing approach. Specifically, each aggregator receives \$2.12 more than they would under the locational pricing approach. This increase corresponds to the budget surplus of \$637, which is collected by the coordinator under the locational pricing approach and is now equally redistributed among all aggregators, aligning with Proposition \ref{proposition1}. These observations confirm that prosumers benefit more under the uniform pricing approach, as DOE limit trading functions as an ancillary service market.
\begin{figure}[!t]
	\centering
	\includegraphics[width=3 in]{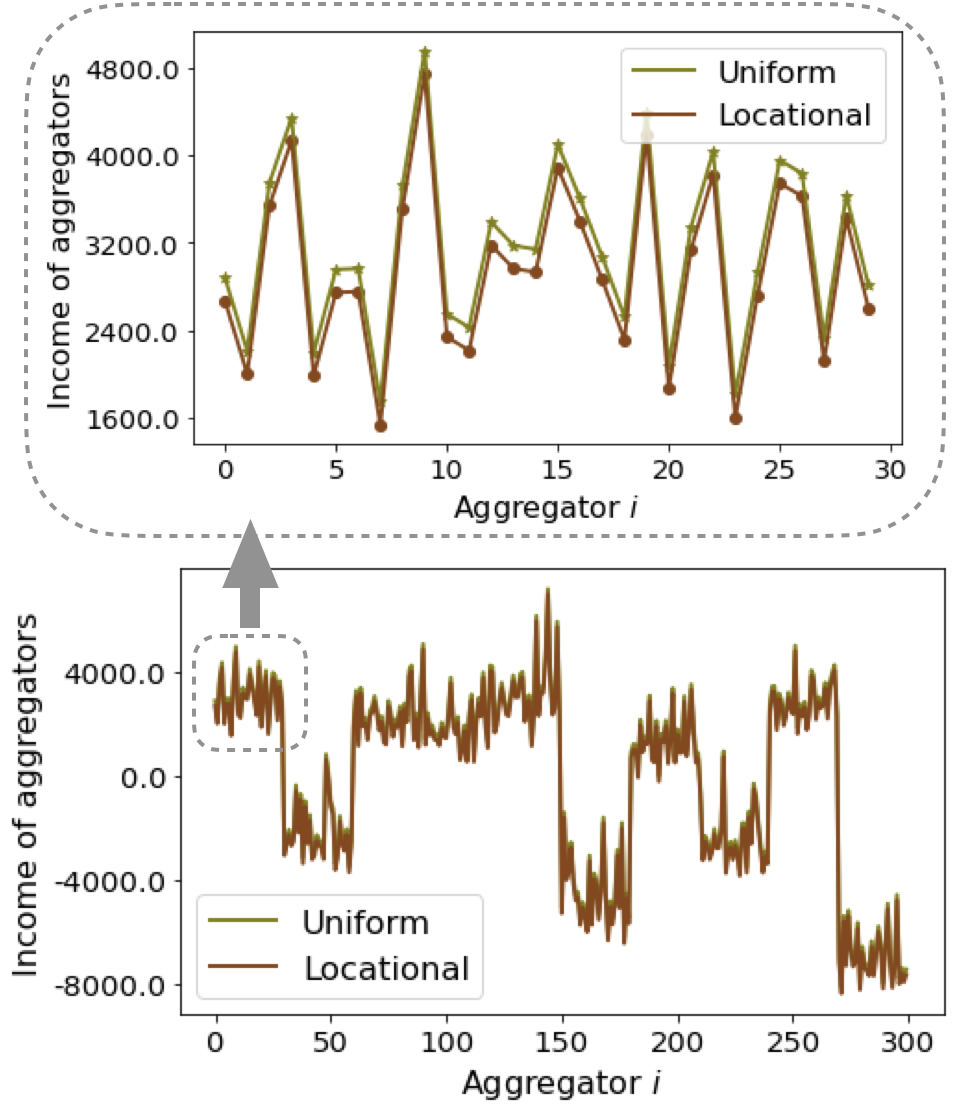}
	\caption{Income of aggregators (in \textcent) in a 24-hour period under the locational and uniform pricing approaches. Aggregators receive a higher income in the uniform pricing approach with DOE limit trading compared to the locational pricing approach. Left: A zoomed-in view of the income of the first 30 aggregators. Each aggregator represents the aggregation of 12 prosumers.}
	\label{fig_income_1}
\end{figure}
\section{Conclusion}\label{sec:conclusion}
In this paper, we proposed P2P energy markets with uniform pricing using the concept of DOEs to address the social acceptance issues arising from locational pricing. By introducing a new degree of freedom for trading DOE limits, we showed that the uniform pricing approach does not compromise grid capacity and is equivalent to the traditional locational pricing approach. Our design is based on competitive markets where participants are price takers. We demonstrated that the resulting market equilibrium, known as the competitive equilibrium, maximizes social welfare, and vice versa. We further investigated the relationship between the proposed uniform pricing markets and the traditional locational pricing market. We showed that, in the locational pricing approach, a budget surplus is generated, whereas, in the uniform pricing approach, the market is strongly budget-balanced, meaning the total payments from participants sum to zero. Finally, we evaluated our methods through a case study of EV charging on an IEEE 13-node test feeder and discussed the effectiveness of the proposed approaches. Future work could explore extensions to systems with uncertainties, where model predictive control potentially addresses challenges related to uncertainties associated with renewable energy generation, demand, and the timing of controllable loads such as EVs.
\section*{Acknowledgment}
The authors would like to express their sincere gratitude to Professor Ying Tan and her research group in the Department of Mechanical Engineering at the University of Melbourne, Australia, for their generous support throughout the completion of this work. 
This research was supported by the Australian Research Council under grants DP190102158, DP190103615, LP210200473, and DP230101014.
\appendices

\section{Proof of Theorem \ref{theorem2}}\label{Appendix_Theorem2}
Let $(\mathbf U^\ast, \mathbf P^\ast)$, along with ${\lambda}^\ast(t)$ for $t \in \mathcal{T}$, be a competitive equilibrium. Merging optimization problems in \eqref{eq_competitive_2} for $i \in \mathcal{N}$ implies that $(\mathbf U^\ast, \mathbf P^\ast)$ maximizes the social welfare in \eqref{eq_social_2}. 
    
    The proof of the reverse direction is based on strong duality, which follows from the satisfaction of Slater's condition. Considering the difference equation \eqref{eq_state}, the state $\mathbf x_i(t)$ evolves according to a linear combination of the initial state $\mathbf x_i(0)$ and the control input $\mathbf U_i$ such that
\begin{equation}\label{eq2_2}
	\mathbf x_i(t)=\mathbf A_i^t \mathbf x_i(0)+ \sum_{j=0}^{t-1}{\mathbf A_i^{t-j-1} \mathbf B_i \mathbf u_i(j)}, \quad t \in \{1, 2, ..., T\}.
\end{equation}
The substitution of \eqref{eq2_2} into  $f_i(\cdot)$ and $\phi_i(\cdot)$ results in
$
		f_i(\mathbf x_i(t), \mathbf u_i(t)) = \tilde{f}_{i,t}(\mathbf U_i)$ and
		$\phi_i(\mathbf x_i(T)) = \tilde{\phi}_i(\mathbf U_i),
$
where functions $\tilde{f}_{i,t}(\cdot)$ and $\tilde{\Phi}_i(\cdot)$ are concave as the composition of a concave function and an affine function. For $i \in \mathcal{N}$, denote $\mathbb{U}_i =\{ \mathbf U_i| \underline{\mathbf u}_i\leq \mathbf u_i(t) \leq \overline{\mathbf u}_i;  \underline{\mathbf x}_i \leq \mathbf x_i(t) \leq \overline{\mathbf x}_i, \text{for } t \in \mathcal{T}\}$,  which is a polyhedral set. For any $(\mathbf U, \mathbf P)$ such that $\mathbf U_i \in \mathbb{U}_i$, the Lagrangian function of \eqref{eq_social_2} is defined as
\begin{equation}\label{eq15_0}
	\begin{aligned}
	&\mathcal L(\mathbf U, \mathbf P,\boldsymbol{\alpha}, \boldsymbol\Psi, \boldsymbol \Pi) = - \sum_{i=1}^{N}\Big(\sum_{t=0}^{T-1} \tilde{f}_{i,t}( \mathbf U_i)+\tilde{\phi}_i(\mathbf U_i) \Big) 
	\\&+ \sum_{t=0}^{T-1}\alpha(t) \Big(-\sum_{i=1}^N p_i(t) \Big) 
	+ \sum_{t=0}^{T-1}\sum_{i=1}^{N} \psi_{i}(t) \Big( p_i(t) +h_i(\mathbf u_i(t)) \\&-a_i(t)\Big)
 + \sum_{t=0}^{T-1}\sum_{i=1}^{N} \boldsymbol \pi_{i}(t) \cdot \Big(\mathbf g_i( p_i(t))-\mathbf w_i^\ast(t)\Big),
\end{aligned}
\end{equation}
where $\psi_{i}(t) \geq 0$, $\boldsymbol \pi_i(t) \geq 0$, $\boldsymbol \alpha =(\alpha(0), \dots, \alpha({T-1}))^\top$,   $\boldsymbol\Psi_i =(\psi_{i}(0), \dots, \psi_{i}(T-1))^\top$, and $ {\boldsymbol\Psi} = ( {\boldsymbol\Psi}_1^\top, \dots,  {\boldsymbol\Psi}_N^\top)^\top$. We define $\boldsymbol\Pi_i$ and  $\boldsymbol\Pi$  in a similar way. 
The Lagrangian  \eqref{eq15_0}  is separable such that
\begin{equation}\label{eq16_0}
	\mathcal L(\mathbf U, \mathbf P, \boldsymbol{\alpha},  \boldsymbol\Psi, \boldsymbol \Pi) = \sum_{i=1}^{N} \mathcal L_i(\mathbf U_i, \mathbf p_i, \boldsymbol{\alpha}, \boldsymbol\Psi_i, \boldsymbol \Pi_i),
\end{equation}
where
$	\mathcal L_i (\mathbf U_i, \mathbf p_i,  \boldsymbol{\alpha},  \boldsymbol\Psi_i, \boldsymbol \Pi_i)
	=- \sum_{t=0}^{T-1} \Big(\tilde{f}_{i,t}(\mathbf U_i)+\alpha(t)  p_i(t) 
	 \Big) 
	+  \sum_{t=0}^{T-1} \psi_{i}(t) \Big( p_i(t)+h_i(\mathbf u_i(t))-a_i(t) \Big) + \sum_{t=0}^{T-1}\boldsymbol \pi_{i}(t) \cdot \Big( \mathbf g_i( p_i(t))-\mathbf w_i^\ast(t)\Big) - \tilde{\phi}_i(\mathbf U_i) $.
Let $(\mathbf U^\star, \mathbf P^\star)$ be an optimal primal solution to  \eqref{eq_social_2}, and  $(-\boldsymbol{\alpha}^\ast,  \boldsymbol{\Psi}^\ast , \boldsymbol{\Pi}^\ast)$ be an optimal dual solution. Given that Slater's condition is satisfied, strong duality leads to 
\begin{equation}\label{eq_primal_4_0}
	(\mathbf U^\star, \mathbf P^\star) \in  \arg \min_{\mathbf U, \mathbf P} \mathcal L(\mathbf U, \mathbf P,  \boldsymbol{\alpha}^\ast,  \boldsymbol\Psi^\ast, \boldsymbol\Pi^\ast),
\end{equation}
for $ \mathbf U_i \in \mathbb{U}_i$. 
Considering \eqref{eq16_0} and \eqref{eq_primal_4_0}, there holds
\begin{equation}\label{eq13_0}
	(\mathbf U_i^\star,  \mathbf p_i^\star) \in \arg \min_{\mathbf U_i, \mathbf p_i} \mathcal L_i( \mathbf U_i,  \mathbf p_i,  \boldsymbol \alpha^\ast, \boldsymbol{\Psi}_i^\ast, \boldsymbol{\Pi}_i^\ast),
\end{equation}
for $ \mathbf U_i \in \mathbb{U}_i$. 
In addition, strong duality implies
 $(\boldsymbol{\alpha}^\ast, \boldsymbol{\Psi}^\ast, \boldsymbol{\Pi}^\ast) \in \arg \max_{\boldsymbol \alpha,  \boldsymbol{\Psi}, \boldsymbol{\Pi}} \min_{\mathbf U, \mathbf P} 	\mathcal L(\mathbf U, \mathbf P,  \boldsymbol{\alpha},  \boldsymbol\Psi, \boldsymbol\Pi) 
	= \arg \max_{\boldsymbol \alpha,\boldsymbol{\Psi}, \boldsymbol{\Pi}} \min_{\mathbf U, \mathbf P} 	\sum_{i=1}^{N} \mathcal L_i(\mathbf U_i, \mathbf p_i,  \boldsymbol{\alpha},  \boldsymbol\Psi_i, \boldsymbol\Pi_i)$ and 
\begin{equation}\label{eq12_0}
	(\boldsymbol{\Psi}_i^\ast, \boldsymbol{\Pi}_i^\ast)  \in  \arg \max_{\boldsymbol \Psi_i, \boldsymbol \Pi_i} \min_{\mathbf U_i, \mathbf p_i}  \mathcal L_i(\mathbf U_i, \mathbf p_i,  \boldsymbol{\alpha}^\ast,  \boldsymbol\Psi_i, \boldsymbol\Pi_i),
\end{equation}
where the primal and dual variables are in their respective domains. Considering  $\lambda^\ast(t)=\alpha^\ast(t)$  for $t \in \mathcal{T}$, the function $ \mathcal L_i(\mathbf U_i, \mathbf p_i,  \boldsymbol{\alpha}^\ast, \boldsymbol\Psi_i, \boldsymbol\Pi_i)$ is the Lagrangian  of \eqref{eq_competitive_2}. 
Therefore,   $(\mathbf \Psi_i^\ast, \mathbf \Pi_i^\ast)$ is an optimal dual solution of \eqref{eq_competitive_2} according to \eqref{eq12_0}.  Given that
Slater's condition is satisfied, strong duality implies that according to  \eqref{eq13_0}, $(\mathbf U^\star, \mathbf P^\star)$ is an optimal primal solution of \eqref{eq_competitive_2}  which satisfies $\sum_{i=1}^N p_i^\ast(t)=0$ for $t\in \mathcal{T}$. 
Consequently, $( \mathbf U^\star, \mathbf P^\star)$, along with $\lambda^\ast(t)=\alpha^\ast(t)$  for $t \in \mathcal{T}$, forms a competitive equilibrium.

\section{Proof of Theorem \ref{theorem3}}\label{Appendix_Theorem3}
The proof is straightforward following the proof of Theorem \ref{theorem2}. Let $( \mathbf U^\ast, \mathbf P^\ast, \mathbf L^\ast)$, along with ${\lambda}^\ast(t)$ and $ \boldsymbol \beta^\ast(t)$ for $t \in \mathcal{T}$,  be a competitive equilibrium. Merging optimization problems in \eqref{eq17} for $i \in \mathcal{N}$ implies that $( \mathbf U^\ast, \mathbf P^\ast, \mathbf L^\ast)$ maximizes the social welfare in \eqref{eq2}. 

   The proof of the reverse direction is based on strong duality, which follows from the satisfaction of Slater's condition. Recall that $\mathbb{U}_i =\{ \mathbf U_i| \underline{\mathbf u}_i\leq \mathbf u_i(t) \leq \overline{\mathbf u}_i;  \underline{\mathbf x}_i \leq \mathbf x_i(t) \leq \overline{\mathbf x}_i, \text{for } t \in \mathcal{T}\}$ for  $i \in \mathcal{N}$,  which is a polyhedral set. For any $(\mathbf U, \mathbf P, \mathbf L)$ such that $\mathbf U_i \in \mathbb{U}_i$, the Lagrangian function of \eqref{eq2} is defined as
\begin{equation}\label{eq15}
	\begin{aligned}
	\mathcal L(&\mathbf U, \mathbf P, \mathbf L, \boldsymbol{\alpha}, \boldsymbol{\delta}, \boldsymbol\Psi, \boldsymbol \Pi) = - \sum_{i=1}^{N}\Big(\sum_{t=0}^{T-1} \tilde{f}_{i,t}( \mathbf U_i)+\tilde{\phi}_i(\mathbf U_i) \Big) 
	\\&+ \sum_{t=0}^{T-1}\alpha(t) \Big(-\sum_{i=1}^N p_i(t) \Big) 
	+ \sum_{t=0}^{T-1} \boldsymbol\delta(t) \cdot \Big(-\sum_{i=1}^N \boldsymbol l_i(t) \Big) 
	\\&+ \sum_{t=0}^{T-1}\sum_{i=1}^{N} \psi_{i}(t) \Big( p_i(t) +h_i(\mathbf u_i(t))-a_i(t)\Big)
 \\&+ \sum_{t=0}^{T-1}\sum_{i=1}^{N} \boldsymbol \pi_{i}(t) \cdot \Big( \boldsymbol l_i(t) +\mathbf g_i( p_i(t))-\mathbf w_i^\ast(t)\Big),
\end{aligned}
\end{equation}
where $\psi_{i}(t) \geq 0$, $\boldsymbol \pi_i(t) \geq 0$, $\boldsymbol \alpha =(\alpha(0), \dots, \alpha({T-1}))^\top$, $\boldsymbol \delta =(\boldsymbol\delta^\top(0), \dots, \boldsymbol\delta^\top({T-1}))^\top$,  $\boldsymbol\Psi_i =(\psi_{i}(0), \dots, \psi_{i}(T-1))^\top$, and $ {\boldsymbol\Psi} = ( {\boldsymbol\Psi}_1^\top, \dots,  {\boldsymbol\Psi}_N^\top)^\top$. We define $\boldsymbol\Pi_i$ and  $\boldsymbol\Pi$  in a similar way. 
The Lagrangian  \eqref{eq15}  is separable such that
\begin{equation}\label{eq16}
	\mathcal L(\mathbf U, \mathbf P, \mathbf L, \boldsymbol{\alpha}, \boldsymbol \delta, \boldsymbol\Psi, \boldsymbol \Pi) = \sum_{i=1}^{N} \mathcal L_i(\mathbf U_i, \mathbf p_i, \mathbf L_i, \boldsymbol{\alpha}, \boldsymbol{\delta}, \boldsymbol\Psi_i, \boldsymbol \Pi_i),
\end{equation}
where
$	\mathcal L_i (\mathbf U_i, \mathbf p_i, \mathbf L_i, \boldsymbol{\alpha}, \boldsymbol{\delta}, \boldsymbol\Psi_i, \boldsymbol \Pi_i)
	=- \sum_{t=0}^{T-1} \Big(\tilde{f}_{i,t}(\mathbf U_i)+\alpha(t)  p_i(t) 
	  + \boldsymbol \delta(t) \cdot \boldsymbol l_i(t)\Big) - \tilde{\phi}_i(\mathbf U_i) 
	+  \sum_{t=0}^{T-1} \psi_{i}(t) \Big( p_i(t)+h_i(\mathbf u_i(t))-a_i(t) \Big) + \sum_{t=0}^{T-1}\boldsymbol \pi_{i}(t) \cdot \Big( \boldsymbol l_i(t) +\mathbf g_i( p_i(t))-\mathbf w_i^\ast(t)\Big)$.
Let $(\mathbf U^\star, \mathbf P^\star, \mathbf L^\star)$ be an optimal primal solution to  \eqref{eq2}, and  $(-\boldsymbol{\alpha}^\ast, -\boldsymbol \delta^\ast, \boldsymbol{\Psi}^\ast , \boldsymbol{\Pi}^\ast)$ be an optimal dual solution. Given that Slater's condition is satisfied, strong duality leads to 
\begin{equation}\label{eq_primal_4}
	(\mathbf U^\star, \mathbf P^\star, \mathbf L^\star) \in  \arg \min_{\mathbf U, \mathbf P, \mathbf L} \mathcal L(\mathbf U, \mathbf P, \mathbf L, \boldsymbol{\alpha}^\ast, \boldsymbol \delta^\ast, \boldsymbol\Psi^\ast, \boldsymbol\Pi^\ast),
\end{equation}
for $ \mathbf U_i \in \mathbb{U}_i$. 
Considering \eqref{eq16} and \eqref{eq_primal_4}, there holds
\begin{equation}\label{eq13}
	(\mathbf U_i^\star,  \mathbf p_i^\star, \mathbf L_i^\star) \in \arg \min_{\mathbf U_i, \mathbf p_i, \mathbf L_i} \mathcal L_i( \mathbf U_i,  \mathbf p_i, \mathbf L_i, \boldsymbol \alpha^\ast, \boldsymbol \delta^\ast, \boldsymbol{\Psi}_i^\ast, \boldsymbol{\Pi}_i^\ast),
\end{equation}
for $ \mathbf U_i \in \mathbb{U}_i$. 
In addition, strong duality implies
\begin{equation*}
 \begin{aligned}
	&(\boldsymbol{\alpha}^\ast, \boldsymbol \delta^\ast, \boldsymbol{\Psi}^\ast, \boldsymbol{\Pi}^\ast) \in \arg \max_{\boldsymbol \alpha, \boldsymbol \delta, \boldsymbol{\Psi}, \boldsymbol{\Pi}} \min_{\mathbf U, \mathbf P, \mathbf L} 	\mathcal L(\mathbf U, \mathbf P, \mathbf L, \boldsymbol{\alpha}, \boldsymbol \delta, \boldsymbol\Psi, \boldsymbol\Pi) 
	\\&= \arg \max_{\boldsymbol \alpha, \boldsymbol \delta, \boldsymbol{\Psi}, \boldsymbol{\Pi}} \min_{\mathbf U, \mathbf P, \mathbf L} 	\sum_{i=1}^{N} \mathcal L_i(\mathbf U_i, \mathbf p_i, \mathbf L_i, \boldsymbol{\alpha}, \boldsymbol{\delta}, \boldsymbol\Psi_i, \boldsymbol\Pi_i),
 \end{aligned}
\end{equation*} 
and, therefore, 
\begin{equation}\label{eq12}
\resizebox {1\linewidth} {!} {$
	(\boldsymbol{\Psi}_i^\ast, \boldsymbol{\Pi}_i^\ast)  \in  \arg \max_{\boldsymbol \Psi_i, \boldsymbol \Pi_i} \min_{\mathbf U_i, \mathbf p_i, \mathbf L_i} \mathcal L_i(\mathbf U_i, \mathbf p_i, \mathbf L_i, \boldsymbol{\alpha}^\ast, \boldsymbol{\delta}^\ast, \boldsymbol\Psi_i, \boldsymbol\Pi_i), $}
\end{equation}
where the primal and dual variables are in their respective domains. Considering  $\lambda^\ast(t)=\alpha^\ast(t)$ and $\boldsymbol \beta^\ast(t)=\boldsymbol \delta^\ast(t)$ for $t \in \mathcal{T}$, the function $ \mathcal L_i(\mathbf U_i, \mathbf p_i, \mathbf L_i, \boldsymbol{\alpha}^\ast, \boldsymbol{\delta}^\ast, \boldsymbol\Psi_i, \boldsymbol\Pi_i)$ is the Lagrangian  of \eqref{eq17}. 
Therefore,   $(\mathbf \Psi_i^\ast, \mathbf \Pi_i^\ast)$ is an optimal dual solution of \eqref{eq17} according to \eqref{eq12}.  Given that
Slater's condition is satisfied, strong duality implies that according to  \eqref{eq13}, $(\mathbf U^\star, \mathbf P^\star, \mathbf L^\star)$ is an optimal primal solution of \eqref{eq17}  which satisfies $
      \sum_{i=1}^N p_i^\ast(t)=0 $ and $\sum_{i=1}^N \boldsymbol l_i^\ast(t)=0$ for $t\in \mathcal{T}
  $.
Consequently, $( \mathbf U^\star, \mathbf P^\star, \mathbf L^\star)$, along with $\lambda^\ast(t)=\alpha^\ast(t)$ and $\boldsymbol \beta^\ast(t)=\boldsymbol \delta^\ast(t)$ for $t \in \mathcal{T}$, forms a competitive equilibrium.

\section{Proof of Theorem \ref{theorem4}}\label{Appendix_Theorem4}
 (i) The proof is based on contradiction. Suppose $(\mathbf U^\star, \mathbf P^\star)$ maximizes \eqref{eq_social_locational}. Select
   $
       \boldsymbol l_i^\star(t)= \mathbf w_i^\ast(t)-\mathbf g_i(p_i^\star(t))+ \frac{1}{N}\big(\sum_{i=1}^N \mathbf g_i(p_i^\star(t))-\sum_{i=1}^N \mathbf w_i^\ast(t)\big)
   $,
   which satisfies $\boldsymbol l_i^\star(t) \leq \mathbf w_i^\ast(t)-\mathbf g_i(p_i^\star(t))$ and $\sum_{i=1}^N \boldsymbol l_i^\star(t)=0$. Select $\mathbf L_i^\star=(\boldsymbol l_i^{\star\top}(0), \cdots, \boldsymbol l_i^{\star\top}(T-1))^\top$ and $\mathbf L^\star=(\mathbf L_1^{\star\top}, \cdots, \mathbf L_N^{\star\top})^\top$. By contradiction, suppose $(\mathbf U^\star, \mathbf P^\star, \mathbf L^\star)$ does not maximize \eqref{eq2}. Then there exists $(\overline{\mathbf U}^\star, \overline{\mathbf P}^\star, \overline{\mathbf L}^\star)$ that satisfies the constraints in \eqref{eq2} and
   \begin{equation}\label{eq3}
\sum_{i=1}^{N}\Big(\sum_{t=0}^{T-1} \tilde{f}_{i,t}( \mathbf U_i^\star)+\tilde{\phi}_i(\mathbf U_i^\star) \Big) < \sum_{i=1}^{N}\Big(\sum_{t=0}^{T-1} \tilde{f}_{i,t}( \overline{\mathbf U}_i^\star)+\tilde{\phi}_i(\overline{\mathbf U}_i^\star) \Big).
   \end{equation}
   It is straightforward to show that $(\overline{\mathbf U}^\star, \overline{\mathbf P}^\star)$ satisfies all constraints in \eqref{eq_social_locational}. Consequently,   \eqref{eq3} contradicts the assumption that $(\mathbf U^\star, \mathbf P^\star)$ maximizes the social welfare in \eqref{eq_social_locational}. Therefore, we conclude that $(\mathbf U^\star, \mathbf P^\star, \mathbf L^\star)$ maximizes \eqref{eq2} as well.

   (ii) The proof of the reverse direction is based on contradiction. Suppose $(\mathbf U^\star, \mathbf P^\star, \mathbf L^\star)$ optimizes \eqref{eq2}. It is straightforward to show that $(\mathbf U^\star, \mathbf P^\star)$ satisfies all constraints in \eqref{eq_social_locational}. By contradiction, suppose $(\mathbf U^\star, \mathbf P^\star)$ does not maximize \eqref{eq_social_locational}. Then there exists another $(\overline{\mathbf U}^\star, \overline {\mathbf P}^\star)$ that satisfies all constraints in \eqref{eq_social_locational} and
 \begin{equation}\label{eq_proof1}
\sum_{i=1}^{N}\Big(\sum_{t=0}^{T-1} \tilde{f}_{i,t}( \mathbf U_i^\star)+\tilde{\phi}_i(\mathbf U_i^\star) \Big) < \sum_{i=1}^{N}\Big(\sum_{t=0}^{T-1} \tilde{f}_{i,t}( \overline{\mathbf U}_i^\star)+\tilde{\phi}_i(\overline{\mathbf U}_i^\star) \Big).
   \end{equation}
   According to part (i), there exists $\overline{\mathbf L}^\star$ which along with $(\overline{\mathbf U}^\star, \overline{\mathbf P}^\star)$ maximizes \eqref{eq2}. Consequently,  \eqref{eq_proof1} contradicts the assumption that $(\mathbf U^\star, \mathbf P^\star, \mathbf L^\star)$ maximizes \eqref{eq2}. Therefore, we conclude that $(\mathbf U^\star, \mathbf P^\star)$ maximizes \eqref{eq_social_locational} as well.

   Following from part (i) and (ii), the two problems are equivalent.

\section{Proof of Theorem \ref{theorem5}}\label{Appendix_Theorem5}
Recall that $\mathbb{U}_i =\{ \mathbf U_i| \underline{\mathbf u}_i\leq \mathbf u_i(t) \leq \overline{\mathbf u}_i;  \underline{\mathbf x}_i \leq \mathbf x_i(t) \leq \overline{\mathbf x}_i, \text{for } t \in \mathcal{T}\}$ for  $i \in \mathcal{N}$,  which is a polyhedral set. For any $(\mathbf U, \mathbf P)$ such that $\mathbf U_i \in \mathbb{U}_i$, 
the Lagrangian function of \eqref{eq_social_locational} is defined as 
\begin{equation}\label{eq5}
	\begin{aligned}
	&\mathcal L(\mathbf U, \mathbf P, \boldsymbol{\alpha}, \underline{\boldsymbol\Xi}, \overline{\boldsymbol\Xi}, \boldsymbol\Psi) = - \sum_{i=1}^{N}\Big(\sum_{t=0}^{T-1} \tilde{f}_{i,t}( \mathbf U_i)+\tilde{\phi}_i(\mathbf U_i) \Big) 
	\\&+ \sum_{t=0}^{T-1}\alpha(t) \Big(-\sum_{i=1}^N p_i(t) \Big) 
	+\sum_{t=0}^{T-1}\sum_{i=1}^{N}\underline\xi_{i}(t) \Big( \underline v_i -v_0 \\&- \sum_{k =1}^N R_{ik}  p_k(t) \Big) 
	+ \sum_{t=0}^{T-1}\sum_{i=1}^{N}\overline\xi_{i}(t) \Big( v_0+ \sum_{k =1}^N R_{ik}  p_k(t)-\overline v_i \Big) 
	\\&+ \sum_{t=0}^{T-1}\sum_{i=1}^{N} \psi_{i}(t) \Big( p_i(t) +h_i(\mathbf u_i(t))-a_i(t)\Big),
\end{aligned}
\end{equation}
where $\underline\xi_{i}(t), \overline\xi_{i}(t), \psi_{i}(t) \geq 0$, $\boldsymbol \alpha =(\alpha(0), \dots, \alpha({T-1}))^\top$, $ \underline{\boldsymbol\Xi}_i =(\underline\xi_{i}(0), \dots, \underline\xi_{i}(T-1))^\top$, and $ \underline{\boldsymbol\Xi} = ( \underline{\boldsymbol\Xi}_1^\top, \dots,  \underline{\boldsymbol\Xi}_N^\top)^\top$. We define $\overline{\boldsymbol\Xi}_i, \overline{\boldsymbol\Xi}$, and $\boldsymbol\Psi_i, \boldsymbol\Psi$ in a similar manner. 
The Lagrangian  \eqref{eq5}  is separable such that 
\begin{equation}\label{eq6}
	\mathcal L(\mathbf U, \mathbf P, \boldsymbol{\alpha}, \underline{\boldsymbol\Xi}, \overline{\boldsymbol\Xi}, \boldsymbol\Psi) = \sum_{i=1}^{N} \mathcal L_i(\mathbf U_i, \mathbf p_i, \boldsymbol{\alpha}, \underline{\boldsymbol\Xi}, \overline{\boldsymbol\Xi}, \boldsymbol\Psi_i),
\end{equation}
where
\begin{equation}\label{eq9}
	\begin{aligned}
	\mathcal L_i &(\mathbf U_i, \mathbf p_i, \boldsymbol{\alpha}, \underline{\boldsymbol\Xi}, \overline{\boldsymbol\Xi}, \boldsymbol\Psi_i)
	=- \sum_{t=0}^{T-1} \Big(\tilde{f}_{i,t}(\mathbf U_i)+\alpha(t)  p_i(t) 
	 \\ &+ \sum_{k=1}^{N} \underline\xi_{k}(t)  R_{ki}  p_i(t)  - \sum_{k=1}^{N} \overline\xi_{k}(t)  R_{ki}  p_i(t) \\&-\underline\xi_{i}(t) (\underline v_i -v_0)  +\overline\xi_{i}(t)  (\overline v_i -v_0) \Big) - \tilde{\phi}_i(\mathbf U_i) 
	\\&+  \sum_{t=0}^{T-1} \psi_{i}(t) \Big( p_i(t)+h_i(\mathbf u_i(t))-a_i(t)\Big).
\end{aligned}
\end{equation}

Let $(\mathbf U^\star, \mathbf P^\star)$ be an optimal primal solution to \eqref{eq_social_locational}, and  $(-\boldsymbol{\alpha}^\ast, \underline{\boldsymbol \Xi}^\ast, \overline{\boldsymbol\Xi}^\ast, \boldsymbol{\Psi}^\ast)$ be an optimal dual solution. Given that Slater's condition is satisfied, strong duality yields 
\begin{equation}\label{eq_primal}
	(\mathbf U^\star, \mathbf P^\star) \in \arg \min_{\mathbf U, \mathbf P} \mathcal L(\mathbf U, \mathbf P, \boldsymbol{\alpha}^\ast, \underline{\boldsymbol\Xi}^\ast, \overline{\boldsymbol\Xi}^\ast, \boldsymbol\Psi^\ast),
\end{equation}
for $ \mathbf U_i \in \mathbb{U}_i$. Denote
\begin{equation}\label{eq10}
	\begin{aligned}
		&\tilde {\mathcal L}_i (\mathbf U_i, \mathbf p_i, \mathbf S_i, \boldsymbol{\alpha}, \underline{\boldsymbol\Xi}, \overline{\boldsymbol\Xi}, \boldsymbol\Psi_i) =  - \sum_{t=0}^{T-1} \Big(\tilde{f}_{i,t}(\mathbf U_i)+\alpha(t)  p_i(t)
		\\ &+ \sum_{k=1}^{N} \underline\xi_{k}(t)  R_{ki}  p_i(t)  - \sum_{k=1}^{N} \overline\xi_{k}(t)  R_{ki}  p_i(t) +\sum_{k=1}^N\underline{  \xi}_k(t)  \underline{ w}_{ki}^\ast(t) \\&+\sum_{k=1}^N \overline{  \xi}_k(t)  \overline{ w}_{ki}^\ast(t)
  -\sum_{k=1}^N\underline{  \xi}_k(t)  \underline{ s}_{ki}(t)-\sum_{k=1}^N \overline{  \xi}_k(t)  \overline{ s}_{ki}(t)\Big)
		\\&+  \sum_{t=0}^{T-1} \psi_{i}(t) \Big( p_i(t)+h_i(\mathbf u_i(t))-a_i(t)\Big)- \tilde{\phi}_i(\mathbf U_i) ,
	\end{aligned}
\end{equation}
where $\mathbf s_i(t)=(\overline{ s}_{1i}(t), \dots, \overline{ s}_{Ni}(t), \underline{ s}_{1i}(t), \dots, \underline{ s}_{Ni}(t))^\top$, $\mathbf S_i=(\mathbf s_i^\top(0), \dots, \mathbf s_i^\top(T-1))^\top$, and $\mathbf S_i \geq 0$.
Considering \eqref{eq6} and \eqref{eq_primal}, we obtain
$
	(\mathbf U_i^\star,  \mathbf p_i^\star) \in \arg \min_{\mathbf U_i, \mathbf p_i} \mathcal L_i( \mathbf U_i,  \mathbf p_i, \boldsymbol \alpha^\ast, \underline{\boldsymbol \Xi}^\ast, \overline{\boldsymbol \Xi}^\ast, \boldsymbol{\Psi}_i^\ast)$.
    Since $\underline\xi_{k}(t)\geq 0$ and $\overline\xi_{k}(t) \geq 0$, there holds 
\begin{equation}\label{eq18}
		(\mathbf U_i^\star,  \mathbf p_i^\star) \in \arg \min_{\mathbf U_i, \mathbf p_i, \mathbf S_i}  \tilde{\mathcal L}_i ( \mathbf U_i,  \mathbf p_i, \mathbf S_i, \boldsymbol \alpha^\ast, \underline{\boldsymbol \Xi}^\ast, \overline{\boldsymbol \Xi}^\ast, \boldsymbol{\Psi}_i^\ast),
\end{equation}
for $ \mathbf U_i \in \mathbb{U}_i$. In addition, strong duality implies
\begin{multline*}
	(\boldsymbol{\alpha}^\ast, \underline{\boldsymbol \Xi}^\ast, \overline{\boldsymbol\Xi}^\ast, \boldsymbol{\Psi}^\ast) \in \arg \max_{\boldsymbol \alpha, \underline{\boldsymbol \Xi}, \overline{\boldsymbol \Xi}, \boldsymbol{\Psi}} \min_{\mathbf U, \mathbf P} 	\mathcal L(\mathbf U, \mathbf P, \boldsymbol{\alpha}, \underline{\boldsymbol\Xi}, \overline{\boldsymbol\Xi}, \boldsymbol\Psi) 
	\\= \arg \max_{\boldsymbol \alpha, \underline{\boldsymbol \Xi}, \overline{\boldsymbol \Xi}, \boldsymbol{\Psi}} \min_{\mathbf U, \mathbf P} 	\sum_{i=1}^{N} \mathcal L_i(\mathbf U_i, \mathbf p_i, \boldsymbol{\alpha}, \underline{\boldsymbol\Xi}, \overline{\boldsymbol\Xi}, \boldsymbol\Psi_i),
\end{multline*}
and
$
	\boldsymbol{\Psi}_i^\ast \in \arg \max_{\boldsymbol \Psi_i} \min_{\mathbf U_i, \mathbf p_i} \mathcal L_i(\mathbf U_i, \mathbf p_i, \boldsymbol{\alpha}^\ast, \underline{\boldsymbol\Xi}^\ast, \overline{\boldsymbol\Xi}^\ast, \boldsymbol\Psi_i)
$.
Consequently, there holds 
\begin{equation}\label{eq7}
	\boldsymbol{\Psi}_i^\ast \in \arg \max_{\boldsymbol \Psi_i} \min_{\mathbf U_i, \mathbf p_i, \mathbf S_i} \tilde {\mathcal L}_i(\mathbf U_i, \mathbf p_i, \mathbf S_i, \boldsymbol{\alpha}^\ast, \underline{\boldsymbol\Xi}^\ast, \overline{\boldsymbol\Xi}^\ast, \boldsymbol\Psi_i),
\end{equation}
where the primal and dual variables are considered to be in their respective domains. The function $\tilde{\mathcal L}_i(\mathbf U_i, \mathbf p_i, \mathbf S_i, \boldsymbol{\alpha}^\ast, \underline{\boldsymbol\Xi}^\ast, \overline{\boldsymbol\Xi}^\ast, \boldsymbol\Psi_i)$ is the Lagrangian  of \eqref{eq17} if we define
\begin{equation}\label{eq11}
\begin{aligned}
 \mathbf g_{it}(p_i(t))&=(R_{1i}, \dots, R_{Ni}, -R_{1i},  \dots, -R_{Ni})^\top p_i(t), 
 \\
  \mathbf{w}^\ast_i(t) &= (\overline w^\ast_{1i}(t), \dots, \overline w^\ast_{Ni}(t), \underline w^\ast_{1i}(t), \dots, \underline w^\ast_{Ni}(t))^\top,\\
  \boldsymbol l_i(t)&= \mathbf{w}^\ast_i(t) - \mathbf g_{it}(p_i(t))- \mathbf s_i(t), \,\,\, \lambda^\ast(t)=\alpha^\ast(t), \\ \boldsymbol \beta^\ast(t)&=(\overline \xi^\ast_1(t), \dots, \overline{\xi}^\ast_N(t), \underline \xi_1^\ast(t), \dots, \underline \xi_N^\ast(t))^\top.
 \end{aligned}
\end{equation}
Consequently,  according to \eqref{eq7}, $\mathbf \Psi_i^\ast$ is an optimal dual solution of \eqref{eq17}.  Besides, since Slater's condition holds, strong duality implies that according to  \eqref{eq18}, $(\mathbf U^\star, \mathbf P^\star)$ is an optimal primal solution of \eqref{eq17}. 
Select the elements of $\mathbf L^\ast$ as
$
       \boldsymbol l_i^\ast(t)= \mathbf w_i^\ast(t)-\mathbf g_i(p_i^\star(t))+ \frac{1}{N}\big(\sum_{i=1}^N \mathbf g_i(p_i^\star(t))-\sum_{i=1}^N \mathbf w_i^\ast (t)\big)$,
   which satisfies $\boldsymbol l_i^\ast(t) \leq \mathbf{w}_i^\ast(t)- \mathbf g_i(p_i^\star(t))$ and $\sum_{i=1}^N \boldsymbol l_i^\ast(t)=0$.
Then, the tuple $( \mathbf U^\star, \mathbf P^\star, \mathbf L^\ast)$, along with $ {\lambda}^\ast(t)$ and $ \boldsymbol \beta^\ast(t)$ for $t \in \mathcal{T}$, forms a competitive equilibrium for \eqref{eq17}. 

\section{Proof of Proposition \ref{proposition_2}}\label{Appendix_Proposition_2}
(i) Consider \eqref{eq9} which is derived from the the Lagrangian  of \eqref{eq_social_locational}.  Comparing the Lagrangian of \eqref{eq_competitive_0} with \eqref{eq9}, the summation of all payments collected by the system coordinator at each time step $t \in \mathcal{T}$ satisfies 
   \begin{equation}\label{eq21_1}
    -\sum_{i=1}^N \lambda_i^\ast(t)p_i^\ast(t)=\sum_{i=1}^N \big[\underline\xi_{i}^\ast(t) (v_0-\underline v_i)  +\overline\xi_{i}^\ast(t)  (\overline v_i -v_0)\big]   
   \end{equation}
   according to the complementary slackness conditions associated with the voltage constraints in \eqref{eq_social_locational}. Recall that the optimal dual variables corresponding to inequality constraints are nonnegative, i.e., $\underline\xi_{i}^\ast(t)\geq 0$ and $\overline\xi_{i}^\ast(t) \geq 0$. By assumption, the lower and upper bounds on the voltage typically represent $\pm5 \% $ of the nominal voltage $v_0$, leading to $(v_0-\underline v_i) \geq 0$ and $(\overline v_i -v_0) \geq 0$. Consequently, it follows that $\sum_{i=1}^N \big[\underline\xi_{i}^\ast(t) (v_0-\underline v_i)  +\overline\xi_{i}^\ast(t)  (\overline v_i -v_0)\big] \geq 0$, and therefore, the summation of all payments collected by the system coordinator is nonnegative.

   (ii) According to the definition of competitive equilibrium, there holds $\sum_{i=1}^N p_i^\ast(t)=0$ and $\sum_{i=1}^N  \boldsymbol l_i^\ast(t)=0$ for $t \in \mathcal{T}$. Then, we have $\sum_{i=1}^N \lambda^\ast(t)p_i^\ast(t)= \lambda^\ast(t) \sum_{i=1}^N p_i^\ast(t)=0$ and $\sum_{i=1}^N \boldsymbol  \beta^\ast(t) \cdot \boldsymbol l_i^\ast(t)= \boldsymbol  \beta^\ast(t) \cdot \sum_{i=1}^N  \boldsymbol l_i^\ast(t)=0$ for $t \in \mathcal{T}$.
   
\section{Proof of Proposition \ref{proposition1}}\label{Appendix_Proposition1}
  Recall that the budget surplus equals \eqref{eq21_1}. If the budget surplus is distributed equally to prosumers, each prosumer receives 
    \begin{equation}\label{eq20_1}
   \frac{1}{N}\sum_{k=1}^N \big[\underline\xi_{k}^\ast(t) (v_0-\underline v_k)  +\overline\xi_{k}^\ast(t)  (\overline v_k -v_0)\big].      
    \end{equation} 
    Note that $v_0-\underline v_i$ and $\overline v_i -v_0$ constitute the elements of vector $\boldsymbol{\nu}(t)$ in \eqref{eq:constraints} when there exists only voltage constraints in the grid. In the meanwhile, the Lagrangian of \eqref{eq17} is provided in \eqref{eq10} when we consider \eqref{eq11}. Comparing \eqref{eq10} with the Lagrangian of \eqref{eq_competitive_0}, we conclude that in the competitive market in \eqref{eq17} each prosumer receives an additional payment of  
    \begin{equation}\label{eq20_2}
    \sum_{k=1}^N \big[\underline{  \xi}_k^\ast(t)  \underline{ w}_{ki}^\ast(t)+ \overline{  \xi}_k^\ast(t)  \overline{ w}_{ki}^\ast(t)\big]    
    \end{equation}
    compared to the competitive market with locational pricing in \eqref{eq_competitive_0}. Comparing \eqref{eq20_1} with \eqref{eq20_2}, there holds $\underline{ w}_{ki}^\ast(t)=\frac{1}{N}(v_0-\underline v_k)$ and $\overline{ w}_{ki}^\ast(t)=\frac{1}{N}(\overline v_k -v_0)$ for $k = \{1, \dots, N\}$. According to \eqref{eq11}, $\overline{ w}_{ki}^\ast(t)$ and $\underline{ w}_{ki}^\ast(t)$ constitute the elements of $\mathbf w_i^\ast(t)$ for $k = \{1, \dots, N\}$. Consequently, it follows that $\mathbf w_i^\ast(t)=\boldsymbol{\nu}(t)/N$.


\begin{IEEEbiography}[{\includegraphics[width=1in,height=1.25in,clip,keepaspectratio]{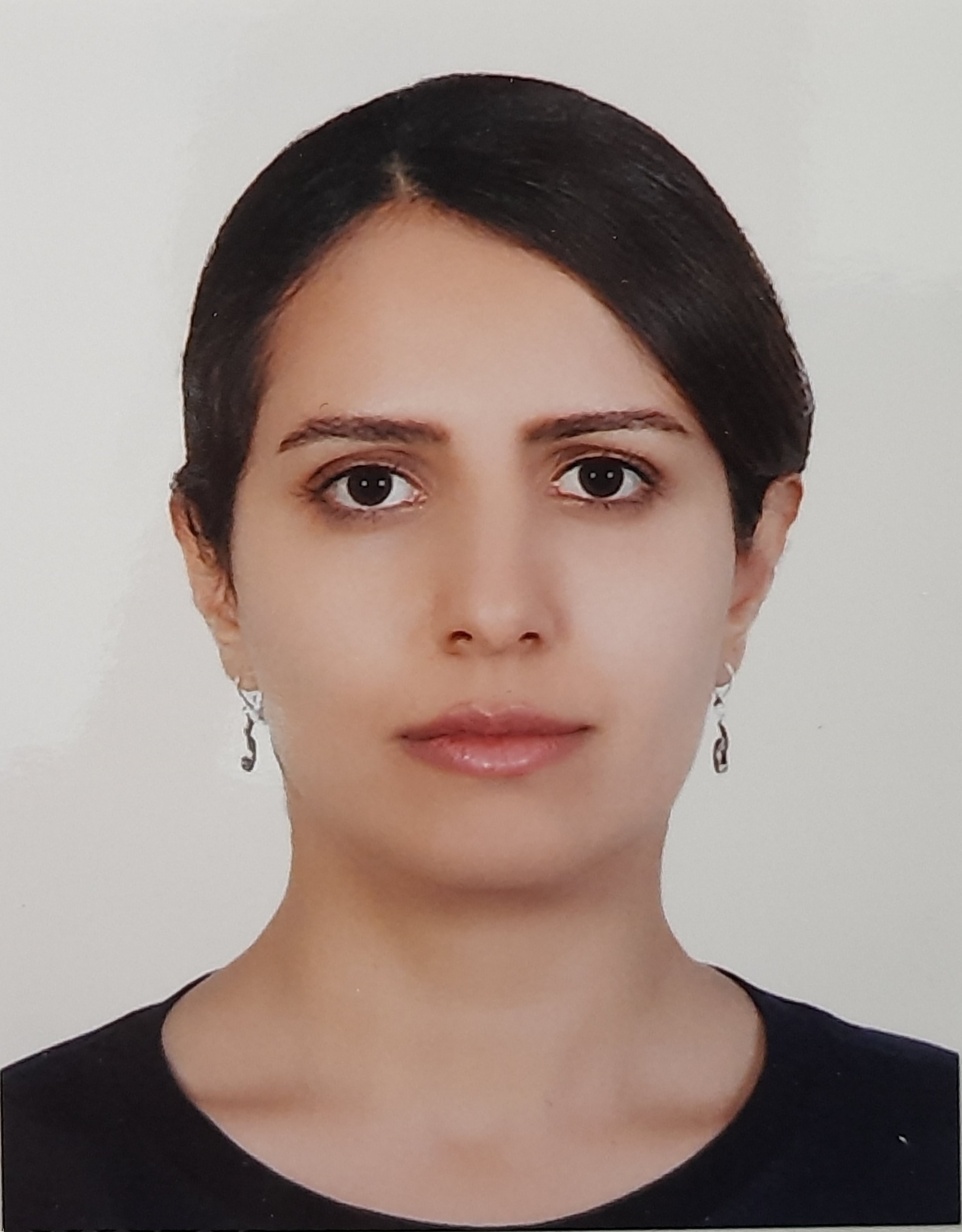}}]
{Zeinab Salehi} received the B.Sc. and M.Sc. degrees
(with Distinction) in electrical engineering from Shiraz University, Shiraz, Iran, in 2015 and
2018, respectively. She is currently working toward the Ph.D. degree in power and control systems engineering with the School of Engineering, the Australian National University,
Canberra, ACT, Australia. Her research interests include control theory and application, multi-agent systems, distributed systems, power system optimization and planning, smart grids, renewable energy integration, machine learning, and model order reduction.
\end{IEEEbiography}


\begin{IEEEbiography}[
{
\includegraphics[width=1in,height=1.25in,clip,keepaspectratio]{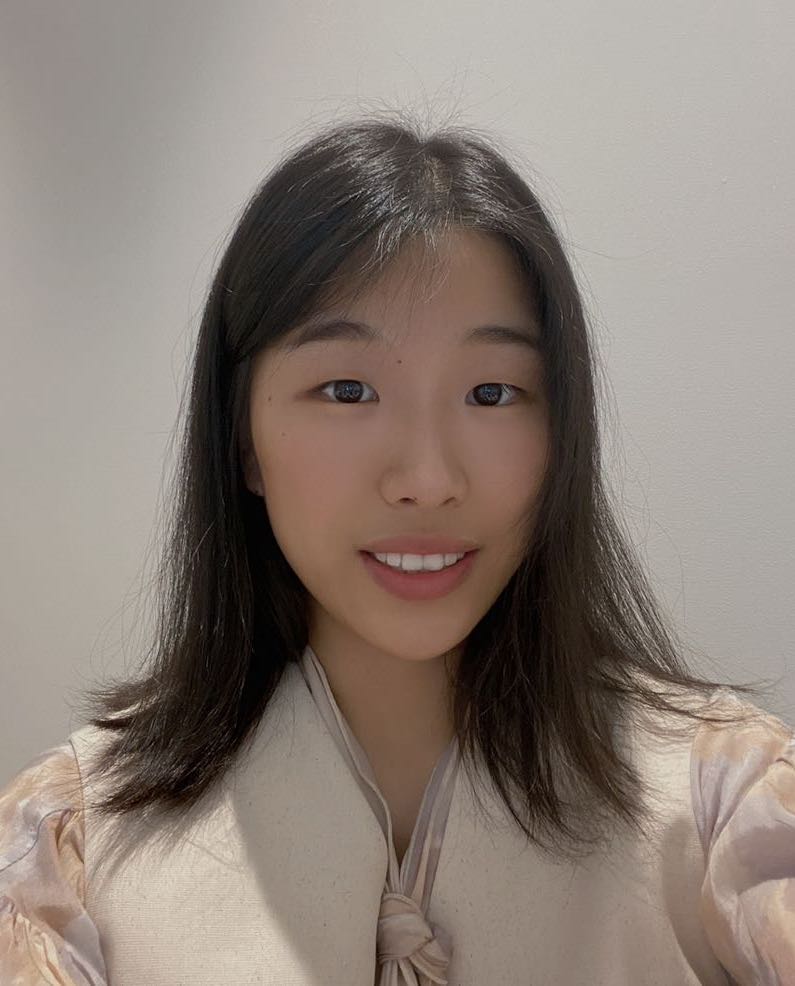}
}]
{Yijun Chen} received her B.Eng. degree in 2019 from the Beijing University of Posts and Telecommunications, China, and her Ph.D. degree in 2023 from the School of Aerospace, Mechanical, and Mechatronic Engineering at the University of Sydney, Australia. Following her Ph.D., she held a postdoctoral position at the Australian National University, where her research focused on enhancing power system stability through the integration of large-scale batteries using negative imaginary systems theory. She is currently a departmental research fellow in the Department of Electrical and Electronic Engineering at the University of Melbourne, with research interests in control theory and power systems. Dr. Chen was recognized as one of the five finalists for the IFAC Congress Young Author Prize in 2023.
\end{IEEEbiography}

\vskip -2\baselineskip plus -1fil

\begin{IEEEbiography}[
{\includegraphics[width=1in,height=1.25in,clip,keepaspectratio]{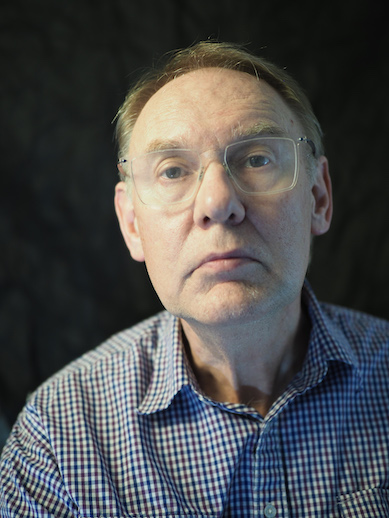}}]
{Ian R. Petersen} was born in Victoria, Australia. He received the Ph.D. degree in electrical engineering from the University of
Rochester, Rochester, NY, USA, in 1984. From 1983 to 1985, he was a Postdoctoral
Fellow with the Australian National University, Canberra, ACT, Australia. From 2017, he has
been a Professor with the School of Engineering, the Australian National University. He was the Interim Director of the School of Engineering with the Australian National University, from 2018–2019. From 1985 until 2016, he was with UNSW Canberra, Canberra, ACT, Australia, where he was a Scientia Professor, an Australian Federation Fellow, and an Australian Research Council Laureate Fellow with the School of Engineering and Information Technology. He was an ARC Executive Director for Mathematics Information and Communications, Acting Deputy ViceChancellor Research with UNSW. His research interests are in robust control theory, quantum control theory and stochastic control theory. Prof. Petersen is a Fellow of the IEEE, the International Federation of Automatic Control and the Australian Academy of Science. He has served as an Editor for \emph{Automatica} and an Associate Editor for \emph{IEEE Transactions on Automatic Control, Systems and Control Letters,
 Automatica, IEEE Transactions on Control Systems Technology},
and \emph{SIAM Journal on Control and Optimization}. 
\end{IEEEbiography}

\vskip -2\baselineskip plus -1fil

\begin{IEEEbiography}[
{\includegraphics[width=1in,height=1.25in,clip,keepaspectratio]{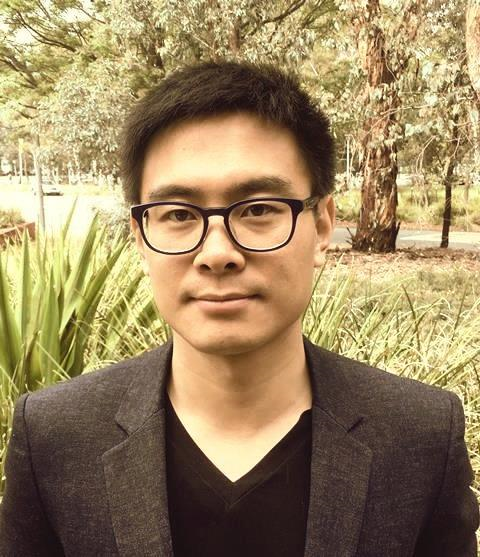}}]
{Guodong Shi} received the B.Sc. degree in mathematics and applied mathematics from the School of Mathematics, Shandong University, Jinan, China, in 2005, and the Ph.D. degree in systems theory from the Academy of Mathematics and Systems Science, Chinese Academy of Sciences, Beijing, China, in 2010. From 2010 to 2014, he was a Postdoctoral Researcher with the ACCESS Linnaeus Centre, KTH Royal Institute of Technology, Stockholm, Sweden. From 2014 to 2018, he was with the Research School of Engineering, The Australian National University, Canberra, ACT, Australia, as a Lecturer/Senior Lecturer, and a Future Engineering Research Leadership Fellow. Since 2019, he has been with the Australian Centre for Robotics, University of Sydney, Camperdown, NSW, Australia. His research interests include distributed control systems, quantum networking and decisions, and social opinion dynamics.
\end{IEEEbiography}

\vskip -2\baselineskip plus -1fil

\begin{IEEEbiography}[
{\includegraphics[width=1in,height=1.25in,clip,keepaspectratio]{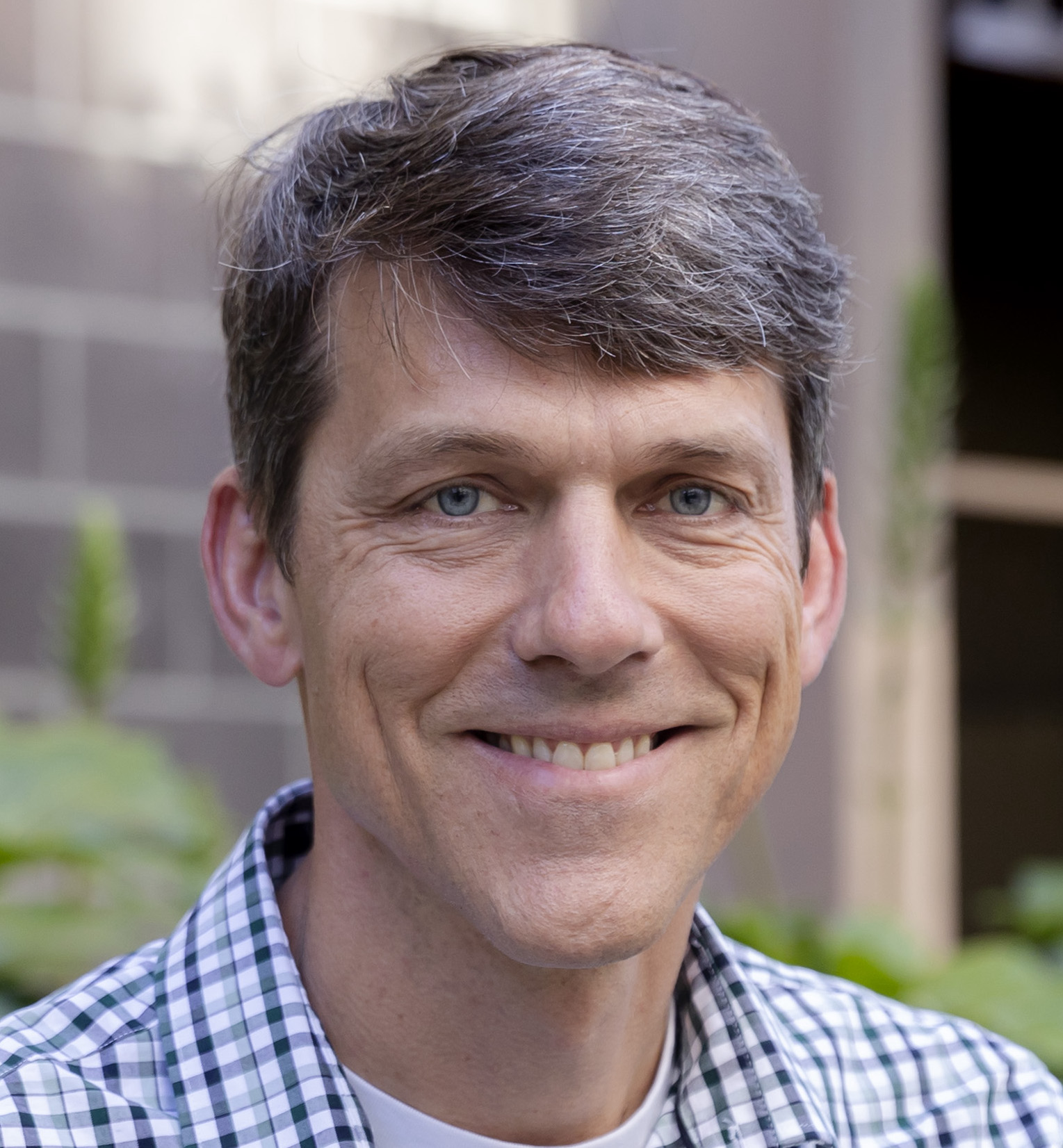}}]
{Duncan S. Callaway} received the
B.S. degree in mechanical engineering from the University of Rochester, Rochester, NY, USA, in 1995,
and the Ph.D. degree in theoretical and applied mechanics from Cornell University, Ithaca, NY, USA.
He is currently an Associate Professor of energy and
resources with University of California, Berkeley,
Berkeley, CA, USA. His current research interests
include control strategies for demand response, electric vehicles and energy storage; distribution network
management; and pathways to electrify and decarbonize the world’s low income regions.
\end{IEEEbiography}

\vskip -2\baselineskip plus -1fil

\begin{IEEEbiography}[
{\includegraphics[width=1in,height=1.25in,clip,keepaspectratio]{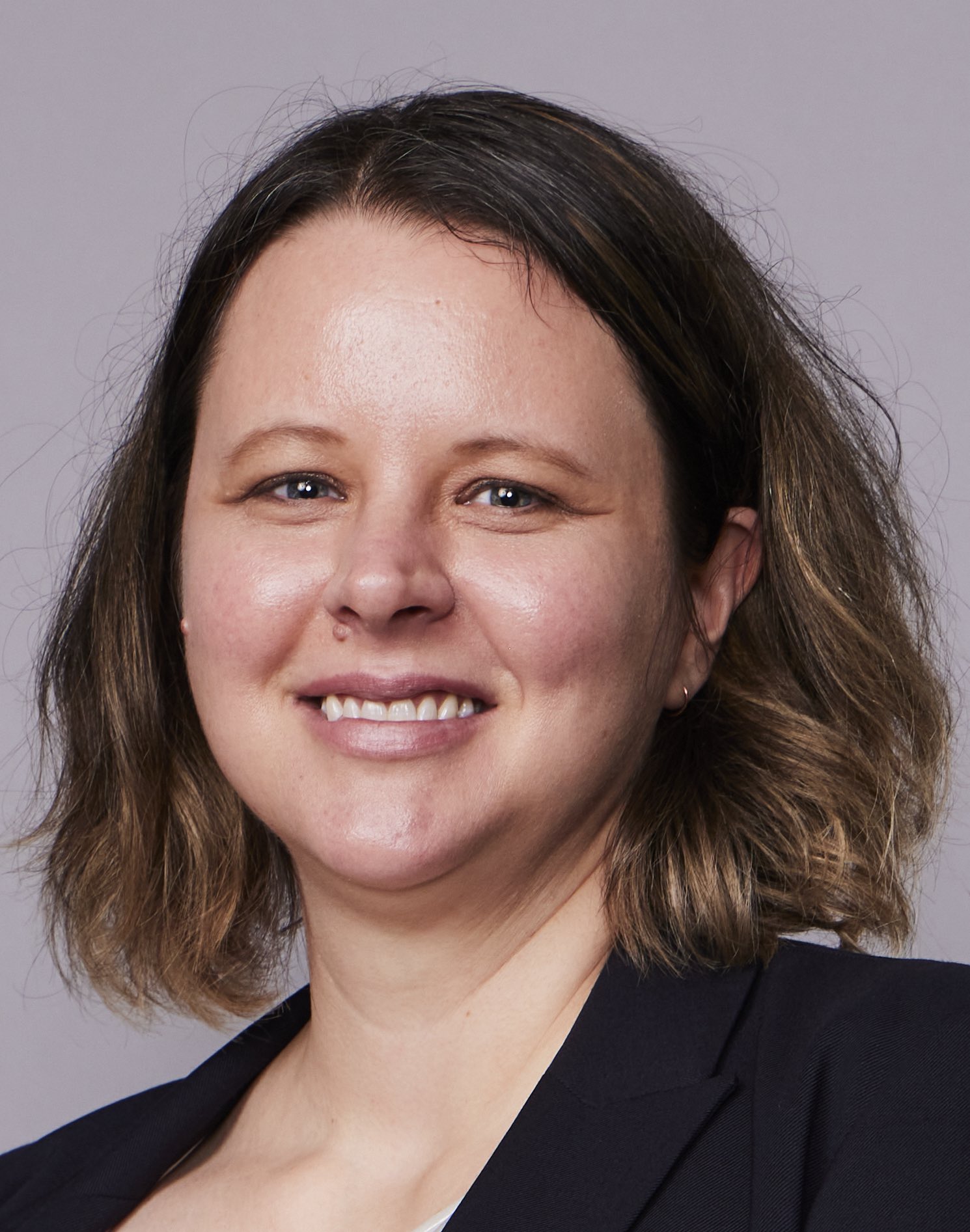}}]
{Elizabeth L. Ratnam} received the B.Eng. (Hons I) degree in electrical engineering and the Ph.D. degree in electrical engineering from the University of Newcastle, Newcastle, NSW, Australia, in 2006 and 2016, respectively. She subsequently held Postdoctoral research positions with the Center for Energy Research, University of California San Diego, San Diego, CA, USA, and the California Institute for Energy and Environment, University of California Berkeley, Berkeley, CA, USA. During 2001–2012, she held various positions with Ausgrid, NSW, Australia, a utility that operates one of the largest electricity distribution networks in Australia. From 2018-2024 she was a Senior Lecturer with the School of Engineering, Australian National University (ANU), Canberra, ACT, Australia, where she was the recipient of the Future Engineering Research Leader Fellowship. Since 2024, Dr. Ratnam is an Associate Professor in the Department of Electrical and Computer Systems Engineering at Monash University. She is a Fellow of Engineers Australia,  and her research interests include robust control of power systems, power system optimization and control, renewable energy grid integration, and smart grids.
\end{IEEEbiography}


\end{document}